\documentclass[letterpaper,journal]{IEEEtran}
\usepackage{cite}
\usepackage{amsmath,amssymb,amsfonts}
\usepackage{graphicx}
\usepackage{textcomp}
\usepackage{xcolor}
\usepackage{float}
\usepackage{booktabs}
\usepackage{xfrac}
\usepackage{bm}
\usepackage{multirow}
\usepackage[colorlinks=true,linkcolor=black!50!blue,citecolor=cyan,urlcolor=blue]{hyperref}
\usepackage{algorithm}
\usepackage{algorithmicx}
\usepackage{algpseudocodex}
\usepackage{yquant}
\useyquantlanguage{groups}
\usepackage[normalem]{ulem}
\usepackage{adjustbox}
\usepackage{makecell}
\usepackage{lscape}
\usepackage{caption}
\def\BibTeX{{\rm B\kern-.05em{\sc i\kern-.025em b}\kern-.08em
    T\kern-.1667em\lower.7ex\hbox{E}\kern-.125emX}}

\renewcommand\a\alpha
\renewcommand\b\beta

\newcommand\<\langle
\renewcommand\>\rangle

\newcommand\ket[1]{\left|#1\right\>}

\DeclareMathOperator{\rank}{rank}

\renewcommand\vec[1]{\bm{#1}}
\renewcommand\v[1]{\vec{#1}}

\newcommand\tensor\otimes
\newcommand\ts\tensor
\newcommand\stp\ltimes
\newcommand\cnot{\mathrm{CNOT}}

\newcommand\Ftwo{\mathbb{F}_2}
\newcommand\supp{\mathrm{supp}}
\newcommand\GL{\mathrm{GL}}

\newcommand\NP{\mathsf{NP}}

\newtheorem{definition}{Definition}

\newtheorem{theorem}{Theorem}

\newtheorem{example}{Example}
\newenvironment{proof}{\noindent\emph{Proof.}}{\hfill$\square$\par}

\newif\ifscaffold
\scaffoldtrue

\ifscaffold
\usepackage[normalem]{ulem}

\newcommand{\branch}[1]{{\color{blue!70!black}\textlangle #1\textrangle}}

\else

\newcommand{\branch}[1]{}

\fi

\newcommand\qiskitsat[1]{{Qiskit-SAT{#1}}}

\usepackage{longtable}
\usepackage{tikz}

\newtheorem{suppdefinition}{Definition}
\newtheorem{supptheorem}{Theorem}[section]
\newtheorem{supplemma}[supptheorem]{Lemma}
\begin{document}

\title{Parallelizable Exact Synthesis of Quantum Circuits via Semi-Tensor Product}

\author{
Chenjian Li, Dingchao Gao, Xiangzhen Zhou*, Ji Guan*, Pengcheng Zhu, Zhufei Chu*\\

%\thanks{Work partially supported by ......
%}

\thanks{
Chenjian Li is with Key Laboratory of System Software (Chinese Academy of Sciences), Institute of Software, Chinese Academy of Sciences, China, and University of Chinese Academy of Sciences, China.

Dingchao Gao is with Key Laboratory of System Software (Chinese Academy of Sciences), Institute of Software, Chinese Academy of Sciences, China, and Beijing Zhongke Arclight Quantum Computing Technology Co., Ltd., China.

Xiangzhen Zhou is with Beijing Zhongke Arclight Quantum Computing Technology Co., Ltd., China, and Nanjing Tech University, China. (email: zhouxiangz@qq.com)

Ji Guan is with Key Laboratory of System Software (Chinese Academy of Sciences), Institute of Software, Chinese Academy of Sciences, China. (email: guanj@ios.ac.cn)

Pengcheng Zhu is with Taizhou University, China.

Zhufei Chu is with Ningbo University, Zhejiang, China. (email: chuzhufei@nbu.edu.cn)
}
}
\maketitle

\begin{abstract}
    Exact synthesis is a key infrastructure in quantum circuit synthesis and optimization, which provides optimal implementations of small circuit shards and is widely used as a circuit re-synthesis optimization kernel. However, existing quantum exact synthesis methods suffer from encoding overhead, memory bottlenecks, and poor parallel scalability. In this work, we introduce a parallel exact synthesis framework for CNOT and phase polynomial circuits based on the semi-tensor product (STP) theory of matrices that avoids these issues. The algorithm contains two stages: it first enumerates candidate circuit topologies, and then instantiates each topology by determining the control and target qubit of its partial gates via a STP-based circuit solver. In the second stage, circuit topologies are encoded as canonical STP expressions, and the CNOT gates are synthesized through right-to-left STP matrix factorization that progressively eliminates infeasible gate decisions. In the framework, topology enumeration and the subsequent solving process are independent across different topologies, and can be naturally parallelized. Despite the $\mathsf{NP}$-hardness of the problem, our algorithm yields up to $12.8\times$ parallel speedup with 32 workers, whereas the parallel speedups of existing SAT-based methods remain below $5\times$ with the same worker budget.
    On randomly generated synthesis targets, the proposed algorithm is typically $100$--$1000\times$ faster than the SAT-based approach on small and moderately difficult instances, and remains competitive for more difficult instances. When integrated in a real-world circuit optimization workflow, our algorithm achieves a median speedup of $3.41\times$ on the QASMBench benchmark.
\end{abstract}


\section{Introduction}
\IEEEPARstart{Q}{uantum} circuit synthesis and optimization is a key infrastructure of the quantum computation stack,  which translates high-level quantum algorithms into quantum circuits that are directly executable on noisy hardware. Since recent quantum computers are still noisy and error-prone~\cite{nisq}, shorter circuits directly translate into lower error rates and higher success rates of quantum algorithms. In particular, two-qubit gates such as CNOT gates are usually much more expensive and error-prone than single-qubit gates~\cite{zuchongzhi,willow,na}, making CNOT reduction a central optimization objective. %General quantum synthesis methods such as cosine-sine decomposition and quantum Shannon decomposition provide asymptotic optimality guarantees for arbitrary unitaries~\cite{csd,qsd}, but they are far from target-specific optimality.

Exact synthesis asks for an implementation with the optimal gate count for a given target, and has been routinely used as a rewriting kernel on small windows in classical~\cite{soeken2018Practical} and quantum~\cite{li2025HOPPS} circuit optimization. Unfortunately, the exact synthesis of quantum circuits is intractable, and is $\NP$-hard even when restricted to CNOT+$R_z$ circuits~\cite{amy2018CNOTcomplexity}. This setting changes the desired algorithmic profile: instead of solving very deep circuits, an exact kernel should have low overhead, certifiable optimality, and be easy to embed in larger synthesis workflows.

Existing exact methods for CNOT circuits fall mainly into two categories: SAT-based methods and database-based methods. SAT-based methods~\cite{shaik2024Optimal,li2025HOPPS} encode the problem into Boolean constraints and invoke a general-purpose SAT solver to find the optimal circuit; they are flexible, but the encoding and tool-invocation overhead can be substantial on small and medium instances, and are difficult to parallelize~\cite{katsirelos2013resolution,haaswijk2020SATBased}. Database-based methods~\cite{CNOT_opt_UCL,CNOT_opt_group_theoretic} enumerate all optimal circuits in advance and answer queries by lookup; they are fast at query time but require humongous memory, making them difficult to scale beyond a few qubits and inconvenient to be integrated into larger synthesis workflows. %In this work, we propose a third approach: an exact circuit solver that preserves and directly exploits the algebraic structure of CNOT circuits.

%This paper introduces a parallel exact-synthesis framework for CNOT and phase polynomial (CNOT+$R_z$) quantum circuits based on the semi-tensor product (STP) theory of matrices~\cite{stp}. The key observation is that CNOT-circuit synthesis can be decomposed into two stages: first, the algorithm determines the locations of interacting qubit pairs, forming \emph{topologies} that serve as CNOT circuit skeletons; in the second stage, a circuit solver based on STP factorization is introduced to determine the control and target qubits of the interacting qubit pairs, thereby instantiating a topology as a concrete CNOT circuit. Both the enumeration stage of circuit skeletons and the solving stage of different skeletons are largely independent across skeletons and can therefore be distributed naturally across workers with little coordination overhead. In our experiments on random circuits with 32 workers, the proposed method achieves a maximum parallel speedup of $12.8\times$, whereas the speedup of the parallel SAT-based method remains below $5\times$ in all tested configurations. The following figure illustrates how a topology is instantiated as a CNOT circuit.

This paper introduces a parallel exact-synthesis framework for CNOT and phase polynomial (CNOT+$R_z$) quantum circuits based on the semi-tensor product (STP) theory of matrices~\cite{stp} that avoids these issues.
Inspired by STP-based circuit solving~\cite{pan2023Exact,pan2023Semitensor} and topology-based exact synthesis~\cite{haaswijk2020SATBased} in classical logic synthesis, we decompose CNOT synthesis into two stages. First, our algorithm determines the circuit \emph{topology}, which consists of interacting qubit pairs that serve as undirected CNOT gates; it then invokes an STP-factorization-based circuit solver to determine the control--target direction of each pair and instantiate the topology into a concrete CNOT circuit. However, adapting this architecture to quantum circuits introduces unique challenges: unlike the typical classical logic gates that has single output, a CNOT gate is a reversible two-output operation, coupling the transformations of both participating qubits. This structural difference requires substantial redesign of the circuit solver. We therefore develop an STP-factorization-based solver tailored specifically to CNOT topology solving. %, achieving orders of magnitude speedups than SAT solvers in solving single topology.
Since different topologies can be enumerated and solved largely independently, both stages can be naturally distributed across workers with little coordination overhead. In our random circuit experiments, the proposed method achieves a maximum parallel speedup of $12.8\times$ with 32 workers, while the parallel SAT-based method remains below $5\times$ across all tested configurations. %The following figure illustrates how a topology is instantiated as a CNOT circuit.

To summarize, the contributions of our work are as follows:
\begin{enumerate}
    \item We propose a two-stage exact synthesis framework for CNOT circuits that combines topology enumeration with an STP-based circuit solver (Sec.~\ref{sec:algorithm}), and further extend the framework to phase polynomial circuits (Sec.~\ref{sec:ext_to_phasepoly}). The proposed synthesis framework avoids heavy encoding overhead, does not rely on large precomputed tables, and naturally exposes parallelism across independent topologies.
    \item We develop an efficient STP-based circuit solver tailored to CNOT synthesis (Sec.~\ref{subsec:factorization_and_solver}). The solver encodes each topology into STP formulas and solves it through symbolic STP factorization, achieving more than $500\times$ speedup over the SAT solver for single-topology solving in most tested configurations (Sec.~\ref{sec:exp_single_topo}).
    \item We develop optimization techniques tailored for CNOT exact synthesis, covering both topology pruning (Sec.~\ref{subsec:topo-prune}) and STP factorization (Theorem \ref{theorem2} in Sec.~\ref{subsec:factorization_and_solver}). The proposed pruning strategies eliminate more than $99.5\%$ candidate topologies in representative cases, while our CNOT-specific simplification substantially reduces the cost of STP factorization.
    \item We evaluate our method on both random synthesis targets (Sec.~\ref{sec:exp_rand_circ}) and real-world circuits (Sec.~\ref{sec:cnot-resynthesis}). STP is typically $100$--$1000\times$ faster than Qiskit-SAT~\cite{qiskit-sat} on small and moderately difficult instances, while remaining competitive on harder cases. On QASMBench~\cite{li2023qasmbench}, the STP-based optimization workflow achieves a median speedup of $3.41\times$ over the SAT-based baseline under a matched worker budget.
\end{enumerate}

%The remainder of this article is organized as follows. In Sec.~\ref{sec:preliminaries} we provide a basic introduction to the related topics and concepts. In Sec.~\ref{sec:algorithm} we introduce our synthesis algorithm and demonstrate how it works for the vanilla CNOT synthesis problem. In Sec.~\ref{sec:ext_to_phasepoly} we extend the synthesis framework to phase polynomial circuits. Finally in Sec.~\ref{sec:experiments} we conduct experiments on random and real-world circuits and demonstrate the effectiveness of our synthesis framework.

\section{Preliminaries}
\label{sec:preliminaries}

\subsection{Quantum Computation}\label{sec:pre_qc}

\paragraph{Qubits and quantum states} Quantum computing is a computing paradigm that attempts to leverage the principles of quantum mechanics to perform computations. While classical data are represented in bits, whose values are \textit{either} 0 or 1, a quantum bit (qubit) can be in a \textit{superposition} of both 0 and 1. Mathematically, a qubit can be described as a linear combination of both 0 and 1:
\begin{equation}
    |\psi\>=\alpha|0\>+\beta|1\>,
\end{equation}
where $\a,\b$ are complex numbers satisfying the normalization constraint $|\a|^2+|\b|^2=1$. Furthermore, an $n$-qubit system can be in a state
\begin{equation}
    \ket\psi=\alpha_{0\cdots 0}\ket{0\cdots 0}+\alpha_{0\cdots 1}\ket{0\cdots 1}+\cdots+\alpha_{1\cdots 1}\ket{1\cdots 1},
\end{equation}
where $\alpha_{\v x}=\a_{x_1x_2\cdots x_n}\in\mathbb{C}$ and $\sum_{\v x}|\alpha_{\v x}|^2=1$. Thus an $n$-qubit quantum state is represented as a $2^n$-dimensional quantum state, exponential in the qubit number $n$. This superposition capability allows qubits to represent exponentially many states simultaneously, which provides a fundamental parallelism advantage over classical bits and is believed to be the root of quantum computation's power.

% \undone{todo}{change this example (no need to use GHZ state! use some more pertinent states instead)}
% \begin{example}
%     The quantum state
%     \begin{equation}
%         |\psi\>=\frac{|000\>+|111\>}{\sqrt 2}
%     \end{equation} is a superposition state over $3$ qubits and is called the GHZ state. %When the qubits are measured, they will end up showing $000$ or $111$ with probability $(\sfrac{1}{\sqrt2})^2=\sfrac 12$ respectively.
% \end{example}

\paragraph{Quantum gates and quantum circuits} Quantum gates are unitary operators (matrices) that map a quantum state to another one, thus acting as the elementary building blocks of quantum computation. Some common quantum gates are:
\begin{enumerate}
    \item $Z$-Rotation ($R_z(\theta)$) gate, which adds a relative phase between computational basis states:
    \begin{equation}
        R_z(\theta)|0\>=e^{-i\theta/2}|0\>,\quad R_z(\theta)|1\>=e^{i\theta/2}|1\>.
    \end{equation}
    \item Hadamard ($H$) gate, which is usually used to create superposition states:
    \begin{equation}
        H|0\>=\frac{|0\>+|1\>}{\sqrt{2}},\quad H|1\>=\frac{|0\>-|1\>}{\sqrt{2}}.
    \end{equation}
    \item Controlled-Not ($\mathrm{CNOT}$) gate, which flips the target qubit when the control qubit is in state $|1\>$:
    % \begin{align}
    %     \cnot\big(|0\>\ts |\psi\>\big)&=|0\>\ts|\psi\>,\,
    %     \cnot\big(|1\>\ts |\psi\>\big)&=|1\>\ts X|\psi\>,
    % \end{align}
    % where the NOT gate, or Pauli-$X$ gate, is defined as
    % \begin{equation}
    %     X|0\>=|1\>,\quad X|1\>=|0\>.
    % \end{equation}

    %CNOT gate is the most widely-used and most important two-qubit gate, and is a key component in many quantum algorithms, such as quantum teleportation and quantum error correction.

    \begin{equation}
        \cnot|x_c,x_t\>=|x_c,(x_c\oplus x_t)\>,
    \end{equation}
    where $x_c,x_t\in\{0,1\}$ denote the control and target bits, respectively, and ``$\oplus$'' denotes addition modulo $2$, i.e., addition over the binary field $\mathbb F_2$. %As a result, a CNOT gate is linear, and admits a $n\times n$ parity matrix representation.
\end{enumerate}

Building upon quantum gates, quantum circuits are diagrams consisting of sequentially applied quantum gates. In quantum circuits, qubits are represented as wires and quantum gates are represented as boxes that are applied from left to right.
% \begin{example}
%     Given two qubits starting at $|00\>$, the following circuit transforms them into $\frac{|00\>+e^{i\theta/2}|11\>}{\sqrt 2}$.
%     \begin{figure}[htbp]
%     \centering
%     \vspace{-1em}
%     \begin{tikzpicture}
%     \begin{yquant}[register/minimum height=4mm, operator/separation=5mm, control style={radius=2pt}]
%         qubit {$|0\>$} q[2];
%         h q[0];
%         box {$R_z(\theta)$} q[0];
%         cnot q[1]|q[0];
%         output {$\frac{|00\>+e^{i\theta}|11\>}{\sqrt 2}$} (-);
%       \end{yquant}
%     \end{tikzpicture}
%     \label{qcirc:GHZ_prepare}
%     \end{figure}
% \end{example}

\paragraph{Linear and phase polynomial circuits}
As the CNOT gate is linear over $\mathbb F_2$, quantum circuits consisting of pure CNOT gates are also called \emph{linear circuits} sometimes. We use the terms CNOT circuits and linear circuits interchangeably throughout the paper. Accordingly, an $n$-qubit CNOT circuit can be represented as an $n\times n$ parity matrix over $\mathbb F_2$.
%, as opposed to a $2^n\times 2^n$ unitary matrix required to represent a general quantum circuit.
For example, a single CNOT gate can be represented as
$\cnot_{1,2}\sim\left[\begin{smallmatrix}
    1 & 0 \\
    1 & 1
\end{smallmatrix}\right].$ A more complicated example is as follows:
\begin{example}\label{example:pmh_linear_circ}
    Consider the following circuit that was first introduced in~\cite{pmh_optimal_linear_synthesis}. Denote the input variables as $x_1\sim x_4$. Then the circuit
    \begin{figure}[H]
    \centering
    \vspace{-0.6em}
    \begin{tikzpicture}
    \begin{yquant}[operator/separation=3mm]
        qubit {} q[4];
        init {$x_1$} q[0];
        init {$x_2$} q[1];
        init {$x_3$} q[2];
        init {$x_4$} q[3];
        cnot q[1]|q[0];
        cnot q[3]|q[2];
        cnot q[2]|q[1];
        cnot q[1]|q[2];
        cnot q[0]|q[1];
        cnot q[3]|q[2];
        output {$x_1\oplus x_3$} q[0];
        output {$x_3$} q[1];
        output {$x_1\oplus x_2\oplus x_3$} q[2];
        output {$x_1\oplus x_2\oplus x_4$} q[3];
      \end{yquant}
    \end{tikzpicture}
    \vspace{-0.6em}
    \label{qcirc:pmh_linear_circ_example}
    \end{figure}implements the following linear parity matrix:
\begin{equation}
    A=\left[\begin{smallmatrix}
        1 & 0 & 1 & 0 \\
        0 & 0 & 1 & 0 \\
        1 & 1 & 1 & 0 \\
        1 & 1 & 0 & 1 \\
    \end{smallmatrix}\right].\label{eq:A_eg}
\end{equation}
\end{example}

Adding $R_z$ gates into pure CNOT circuits, one can further form phase polynomial circuits~\cite{dawson2004quantum,montanaro2016quantum,chen2025phasepoly}. This type of circuit can be expressed in a \emph{sum-over-path} form:
\begin{equation}
    U\ket{\v x} = e^{2\pi i p(\v x)}\ket{A\v x},
\end{equation}
where $p(\v x)$ is the \emph{phase polynomial}. Here $A$ is the parity matrix corresponding to the CNOT circuit skeleton, and $p(\v x)$ is used to represent the extra phases introduced by the $Z$-rotation gates. Specifically, $p(\v x)$ usually takes the following form~\cite{t-par,amy2018CNOTcomplexity}:
\begin{equation}
    p(x_1,\cdots,x_n)=\sum_{\v a\in \{0,1\}^n}\theta_{\v a}\cdot \big(a_1x_1\oplus\cdots\oplus a_nx_n\big).
\end{equation}
Following Amy's notation~\cite{amy2018CNOTcomplexity}, the terms $a_1x_1\oplus \cdots \oplus a_nx_n$ can be represented as
\begin{equation}
    p_{\v a}(\v x)=\bigoplus_{i=1}^n a_ix_i
\end{equation}
and are referred to as \emph{parities}.

\begin{example}
    The phase polynomial of the $CCZ$ gate, or the doubly-controlled-$Z$ gate, is as follows:
    \begin{multline}
        p(\v x)=\frac{1}{2}x_1\cdot x_2\cdot x_3\equiv \frac{1}{8}\Big(x_1+x_2+x_3+ 3(x_1\oplus x_2) \\+3(x_2\oplus x_3)+3(x_1\oplus x_3)+x_1\oplus x_2\oplus x_3\Big)\, (\text{mod }1).
    \end{multline}
    Furthermore, the $CCZ$ gate can be implemented with the following circuit:
    \begin{figure}[H]
    \centering
    \vspace{-.6em}
    \begin{adjustbox}{scale=0.8, center}
    \begin{tikzpicture}
    \begin{yquant}[operator/separation=3mm]
        qubit {} x[3];
        init {$x_1$} x[0];
        init {$x_2$} x[1];
        init {$x_3$} x[2];
        box {$R_z(\frac{\pi}{4})$} x[0];
        box {$R_z(\frac{\pi}{4})$} x[1];
        cnot x[0]|x[2];
        cnot x[2]|x[1];
        box {$R_z(\frac{3\pi}{4})$} x[0];
        box {$R_z(\frac{3\pi}{4})$} x[2];
        cnot x[0]|x[1];
        cnot x[2]|x[1];
        box {$R_z(\frac{\pi}{4})$} x[0];
        cnot x[0]|x[2];
        box {$R_z(\frac{3\pi}{4})$} x[0];
        box {$R_z(\frac{\pi}{4})$} x[2];
        cnot x[0]|x[1];
      \end{yquant}
    \end{tikzpicture}
\end{adjustbox}
    \vspace{-1.8em}
    \label{qcirc:CCZ}
    \end{figure}
\end{example}

\paragraph{Partial gates and circuit topologies} \label{para:pre_partial_gates_topo} In our synthesis framework, we need to represent CNOT gates whose direction is undetermined; leveraging the nomenclature from the classical logic synthesis community, we call such undetermined gates (2-qubit) \textit{partial gates}, which are represented using the following notation:

\begin{figure}[H]
\centering
\vspace{-1em}
\begin{tikzpicture}
  \begin{yquantgroup}[register/minimum height=4mm, operator/separation=1mm]
     \registers{
       qubit {} q[2];
     }
     \circuit{
       [control style={draw,fill=white,shape=yquant-rectangle,radius=3pt}, operator style={draw,fill=white,shape=yquant-rectangle,radius=3pt}]
       cnot q[0]|q[1];
     }
     \equals
     \circuit{
       cnot q[1]|q[0];
     }
     \equals[or]
     \circuit{
       cnot q[0]|q[1];
     }
  \end{yquantgroup}
\end{tikzpicture}
\vspace{-1em}
\end{figure}
The direction of a partial gate is determined by a single binary variable that takes the value $\mathsf{CT}$ (the left case in the figure) or $\mathsf{TC}$ (the right case). We call the variable a \textit{gate direction} or a \textit{gate decision}.
Furthermore, a quantum circuit composed of partial gates is called a  \textit{(circuit) topology}. A partial gate or circuit topology is said to be \textit{instantiated} into an actual gate or circuit when gate decisions are assigned to it.

More generally, following previous works on classical exact synthesis~\cite{soeken2018Practical}, one can introduce multi-qubit partial gates, which span across multiple qubits, and is later instantiated by selecting two distinct qubits as the control and target qubit. For example, the following figure demonstrates a 3-qubit partial gate with $3\times 2=6$ possible gate instantiations.

\begin{figure}[H]
\centering
\vspace{-.6em}
\begin{adjustbox}{max width=\columnwidth,center}
\begin{tikzpicture}
  \begin{yquantgroup}[register/minimum height=4mm, operator/separation=1mm]
     \registers{
       qubit {} q[3];
     }
     \circuit{
       [control style={draw,fill=white,shape=yquant-rectangle,radius=3pt},
        operator style={draw,fill=white,shape=yquant-rectangle,radius=3pt}]
       cnot q[0]|q[1,2];
     }
     \equals
     \circuit{
       cnot q[1]|q[0];
       align q;
     }
     \equals[or]
     \circuit{
       cnot q[0]|q[1];
       align q;
     }
     \equals[or]
     \circuit{
       cnot q[2]|q[0];
       align q;
     }
     \equals[or]
     \circuit{
       cnot q[0]|q[2];
       align q;
     }
     \equals[or]
     \circuit{
       cnot q[2]|q[1];
       align q;
     }
     \equals[or]
     \circuit{
       cnot q[1]|q[2];
       align q;
     }
  \end{yquantgroup}
\end{tikzpicture}
\end{adjustbox}
\vspace{-1.8em}
\label{fig:three_qubit_partial_gate}
\end{figure}

However, we will later show in Sec.~\ref{subsec:multi-qubit_partial_gate} that is optimal in our application, in the sense that it introduces least decision space redundancy and is relatively easy to parallelize. We therefore use 2-qubit partial gates throughout the paper unless otherwise stated, and simply refer to them as partial gates whenever no ambiguity arises.

\subsection{Semi-Tensor Product}

The semi-tensor product of matrices generalizes conventional matrix multiplication to matrices with mismatched dimensions. Interested readers can refer to \emph{An Introduction to Semi-Tensor Product of Matrices and Its Applications}~\cite{stp} for a more detailed introduction.
\begin{definition}[Semi-Tensor Product~\cite{stp}]
    For two matrices $A\in \mathbb{F}^{m\times k_1}$ and $B\in \mathbb{F}^{k_2\times n}$ where $\mathbb{F}$ denotes an arbitrary field, their semi-tensor product is defined as
\begin{equation}
    A\stp B = (A\otimes I_{K/k_1})\cdot (B\otimes I_{K/k_2}),
\end{equation}
where $K:=\mathrm{lcm}(k_1,k_2)$.
\end{definition}

When $k_1=k_2$, STP reduces to ordinary matrix multiplication. When the dimensions mismatch, the inserted identity matrices act as placeholders that carry unaffected variables through the product.

% As a first and direct application, STP can be used to represent the combination of quantum systems. Suppose quantum systems are represented as vectors
% \begin{equation}
%     |q_1\>=\begin{bmatrix}
%             \a \\ \b
%     \end{bmatrix},
%     |q_2\>=\begin{bmatrix}
%         \gamma \\ \delta
%     \end{bmatrix},\\
% \end{equation}
% then the combination of quantum systems happens to be the semi-tensor product of two vector states:
% \begin{equation}
%     \begin{bmatrix}
%         \a \\ \b
%     \end{bmatrix}\stp \begin{bmatrix}
%         \gamma \\ \delta
%     \end{bmatrix} =\left(\begin{bmatrix}
%         \a \\ \b
%     \end{bmatrix}\otimes I_2\right)\cdot \begin{bmatrix}
%         \gamma \\ \delta
%     \end{bmatrix} =
%     \begin{bmatrix}
%         \a \\ \b
%     \end{bmatrix}\otimes \begin{bmatrix}
%         \gamma \\ \delta
%     \end{bmatrix}=|q_1q_2\>.
% \end{equation}

STP can be used to model multivariate Boolean functions. Denote logical True and False by vectors:
    \begin{equation}
    \label{eq:tf}
        \mathsf{False}=|0\>=\begin{bmatrix}
            1 \\ 0
        \end{bmatrix},
        \mathsf{True}=|1\>=\begin{bmatrix}
            0 \\ 1
        \end{bmatrix}.
    \end{equation}
Then any binary logic operation $\sigma$ can be represented as a STP formula:
    \begin{equation}
        a\ \sigma\ b= M_{\sigma}\stp a \stp b,
    \end{equation}
    where $a,b$ are Boolean variables (in their vector form), and $M_\sigma$ is the \textit{structural matrix} corresponding to the Boolean operation $\sigma$. For example, for $\sigma=\oplus$, we have
    \begin{equation}
        a\oplus b=M_{\oplus}\stp a\stp b,
    \end{equation}
    where $M_\oplus =\left[\begin{smallmatrix}
        1 & 0 & 0 & 1 \\
        0 & 1 & 1 & 0 \\
    \end{smallmatrix}\right]$.

\begin{figure*}[t] % 手调位置
    \centering
    \includegraphics[width=0.98\textwidth]{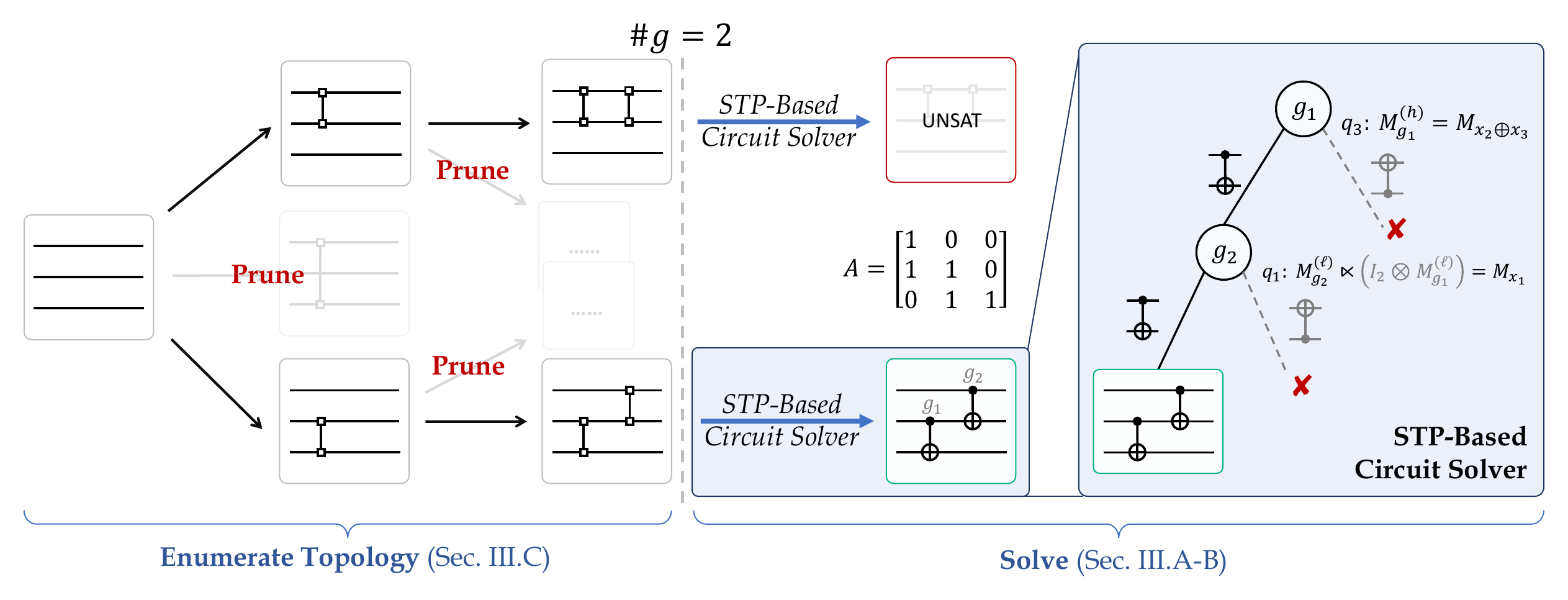}
    \caption{Overall structure of our STP-based synthesis algorithm for a fixed gate count ($\# g$). Given a synthesis target, candidate topologies are enumerated, while infeasible ones are eliminated through pruning.
    Each surviving topology is then passed to an STP-based circuit solver, which either proves it unsatisfiable or returns a complete circuit implementation. Both topology enumeration and circuit solving can be processed in parallel, since different topologies can be handled independently.}
    \label{fig:algo_arch}
\end{figure*}

The STP expression for a general logic formula can be complicated. For example, the logic expression
\begin{equation}
    \phi(a,b,c)=c\leftrightarrow( a \land \lnot b)
\end{equation}
can be converted into STP form
\begin{equation}
    \phi(a,b,c)= M_{\leftrightarrow} c \big(M_\land a (M_\lnot b)\big),\label{eq:logic_stp}
\end{equation}
where all ``$\stp$'' symbols have been omitted for simplicity. Furthermore, any STP expression can be converted into a simple \textit{canonical form}, which is defined as
\begin{equation}
    \Phi(x_1,x_2,...,x_n)=M_{\Phi}\stp x_1\stp x_2\cdots \stp x_n,
\end{equation}
where $M_\Phi$ is a single $2\times 2^n$ matrix. Some simple reasoning can show that the following three rules suffice to convert any STP expression to its canonical form:

\begin{enumerate}
    \item (Matrix-Variable Swap Rule) For a matrix $M$ and a variable $a$, $a\stp M=(I_2\otimes M)\stp a$.
    \item (Variable-Variable Swap Rule) For variables $a, b$, $b\stp a=M_w\stp a\stp b$, where
    \begin{equation}
        M_w=\mathsf{SWAP}=\left[\begin{smallmatrix}
            1 & 0 & 0 & 0 \\
            0 & 0 & 1 & 0 \\
            0 & 1 & 0 & 0 \\
            0 & 0 & 0 & 1 \\
        \end{smallmatrix}\right]
    \end{equation}
    is the swap matrix.
    \item (Variable-Reduce Rule) For a binary variable $a$, $a^{\stp 2}=a\stp a=M_r\stp a$, where
    \begin{equation}
        M_r=\left[\begin{smallmatrix}
            1 & 0 \\
            0 & 0 \\
            0 & 0 \\
            0 & 1 \\
        \end{smallmatrix}\right]
    \end{equation}
    is the power-reducing matrix.
\end{enumerate}

\begin{example}
    Convert~\eqref{eq:logic_stp} to its canonical form. Omit the STP symbols for simplicity.
    \begin{align}
        &\mathrel{\phantom{=}}\phi(a,b,c)\\
        &=M_{\leftrightarrow} c \big(M_\land a (M_\lnot b)\big)\\
        &=M_{\leftrightarrow}(I_2\otimes M_\land)(I_4\otimes M_{\lnot})cab \\
        &=M_{\leftrightarrow}(I_2\otimes M_\land)(I_4\otimes M_{\lnot})M_w(I_2\otimes M_w)abc.
    \end{align}
\end{example}

\section{STP Exact Synthesis Framework}
\label{sec:algorithm}

%Core techniques: Factorization, Decision Forest
In this section, we introduce our synthesis framework. We first describe the overall structure; then give a detailed illustration of each module. The overall algorithm structure is illustrated in Fig.~\ref{fig:algo_arch}, and the detailed pseudocode for our algorithm is listed in Alg.~\ref{alg:exact_syn} and Alg.~\ref{alg:factor}.

Our algorithm takes an $n$-qubit linear parity matrix $A\in\mathrm{GL}(n,2)$ as input. It enumerates all topologies with $1$ partial gate, $2$ partial gates, and so on, until one or more solutions are found. Since topologies are explored in increasing gate count, the first feasible solutions are guaranteed to be optimal.% as opposed to heuristic or A* algorithms~\cite{qsearch,CNOT_opt_UCL}.
For each topology, a STP-factorization-based circuit solver is then called to determine whether the partial gates can be instantiated to implement the target $A$. As there are no dependencies between topologies, different topologies can be processed independently and hence solved directly in parallel. Since the number of candidate topologies grows exponentially as the number of partial gates increases, we further develop a group of effective topology-pruning strategies to eliminate infeasible candidates before solving to improve the algorithm's efficiency (Sec.~\ref{subsec:topo-prune}).

To solve a topology, we first convert it into canonical STP expressions and constrain their outputs according to the target matrix $A$ (Sec.~\ref{subsec:topo_to_stp}). The circuit solver then factorizes these expressions from right to left, eliminating gate decisions whenever the corresponding factorization becomes infeasible, until all constraints are resolved (Sec.~\ref{subsec:factorization_and_solver}). If all gate decisions are ruled out, the solver terminates early and returns $\mathsf{UNSAT}$; otherwise, the surviving gate decisions instantiate the topology into valid circuit solutions.

Throughout the above framework, we use the two-qubit partial gates. But partial gates may actually span $k>2$ qubits in general, as shown in classical exact synthesis~\cite{soeken2018Practical}. After presenting the main algorithmic modules, we revisit this design choice in Sec.~\ref{subsec:multi-qubit_partial_gate} and show that the efficiency of the topology-based synthesis algorithm is more likely to be optimal at the two extremal choices, $k=2$ and $k=n$ for $n$-qubit synthesis target, by analyzing and the topology-to-circuit representation redundancy. The analysis provides a theoretical justification for our use of 2-qubit partial gates in the framework and also identifies the $k=n$ architecture (which is more difficult to parallelize) as an alternative direction for future work.

%Throughout the above framework, we use the two-qubit partial gates. But as shown in classical exact synthesis~\cite{soeken2018Practical}, partial gates may generally span $k>2$ qubits in general. After presenting the main algorithmic modules, we revisit this design choice in Sec.~\ref{subsec:multi-qubit_partial_gate} and analyze how the partial-gate width $k$ affects the overall synthesis search space. Under a simplified solver model, we show that the overall synthesis search space is minimized at the two extremal choices, $k=2$ and $k=n$ for $n$-qubit synthesis target, whereas intermediate widths introduce redundant representations of the same CNOT gates. This analysis provides a theoretical justification for our use of 2-qubit partial gates adopted in our framework and also identifies the $k=n$ architecture (which is more difficult to parallelize) as an alternative direction for future work.

\subsection{Encoding topology as STP formulas}\label{subsec:topo_to_stp}
For each topology, the solver starts by converting the topology to STP expressions. Consider for example, the following topology and target:
\begin{figure}[H]
\centering
\vspace{-.6em}
\begin{tikzpicture}
  \begin{yquant}[register/minimum height=4mm, operator/separation=4mm]
    qubit {$x_{\protect\the\numexpr\idx+1}$} q[3];
    [control style={draw,fill=white,shape=yquant-rectangle,radius=3pt}, operator style={draw,fill=white,shape=yquant-rectangle,radius=3pt}]
    cnot q[1]|q[2];
    [control style={draw,fill=white,shape=yquant-rectangle,radius=3pt}, operator style={draw,fill=white,shape=yquant-rectangle,radius=3pt}]
    cnot q[0]|q[1];
    %output {$x_{\protect\the\numexpr\idx+1}''$} q;
    output {$x_1''=x_1$} q[0];
    output {$x_2''=x_1\oplus x_2$} q[1];
    output {$x_3''=x_2\oplus x_3$} q[2];
  \end{yquant}
  %\node[] at (1.1,0) {$q_1'$};
  \node[] at (1.1,-0.46) {$x_2'$};
  %\node[] at (1.1,-1.2) {$q_3'$};
  \node[] at (0.7,-1.5) {$g_1$};
  \node[] at (1.6,-1.5) {$g_2$};
\end{tikzpicture}
\vspace{-1.2em}
\end{figure}

For the first output qubit, the canonical STP expression of the topology is
\begin{align}
    \Phi_{T,1} x_1x_2x_3&=M_{g_2}^{(\ell)}x_1x_2'\\
    &=M_{g_2}^{(\ell)}x_1 M_{g_1}^{(\ell)}x_2x_3\\
    &=M_{g_2}^{(\ell)}(I_2\ts M_{g_1}^{(\ell)})x_1x_2x_3,
\end{align}
where $M_{g_i}$ denotes the structural matrix of partial gate $g_i$, and the superscripts $(\ell)$ and $(h)$ distinguish its two outputs. Here we have omitted the ``$\stp$'' symbol for readability. Meanwhile, the target for the first output qubit is $x_1''=x_1$, so the required target structural matrix for the first output qubit is:
\begin{equation}
    x_1''=M_1 x_1 x_2 x_3,\quad
    M_1=\left[\begin{smallmatrix}
        1 & 1 & 1 & 1 & 0 & 0 & 0 & 0 \\
        0 & 0 & 0 & 0 & 1 & 1 & 1 & 1 \\
    \end{smallmatrix}\right].
\end{equation}
Since the topology must realize the target for every input, their structural matrices must be equal, yielding the constraint $\Phi_{T,1}=M_1$. Then repeating the same construction for every output qubit gives the complete constraint system for topology $T$, which is then passed to the solver:
\begin{equation}
    \Phi_{T,i}=M_i, \quad \forall\, i\in[n].
\end{equation}

\begin{algorithm}[t]
	\caption{Exact synthesis algorithm for CNOT circuits}
    \label{alg:exact_syn}
	\begin{algorithmic}[1]
		\Require Linear parity matrix $A\in\mathrm{GL}(n,2)$.
		\Ensure Optimal CNOT circuit(s) implementing $A$.
		\Function {ExactSynLinearSTP}{$A$}
			\State $ \#g\gets 0$ \Comment{Target gate number}
            \State $S\gets \varnothing$ \Comment{Solution set}
            \While{$S= \varnothing$}
                \State $ \mathcal T_{\#g,A}\gets\textsc{EnumerateTopology}(\#g, A)$
                \For{$ T\in \mathcal T_{\#g, A}$} \Comment{Parallelizable}
				    %\If{$\textsc{Solve}(T, A)\neq\mathsf{UNSAT}$}
                    \State $\mathcal D\gets \textsc{Solve}(T,A)$
                    \State $S\gets S\cup\{\textsc{Instantiate}(T, \mathcal D)\} $
				    %\EndIf
                \EndFor
                \State $\#g\gets \#g +1$
            \EndWhile
            \State \Return $S$
		\EndFunction
        \vspace{.6em}
        \Function {EnumerateTopology}{$\#g, A$}
			\State $\mathcal T_{0,A}\gets\{C_\epsilon\}$ \Comment{$n$-qubit empty circuit}
            \For{$\#g'=1,2,\cdots,\#g$}
                \State $\mathcal T_{\#g',A}\gets\varnothing$
                \For{$T\in \mathcal T_{\#g'-1,A}$} \Comment{Parallelizable}
                    \For{$q_1,q_2\in[n]$ with $q_1<q_2$}
                        \State $T_{\mathrm{new}}\gets T:\mathrm{PartialGate}(q_1,q_2)$
                        \If{not \textsc{Prune}($T_{\mathrm{new}},\#g,A$)}
                            \State $\mathcal T_{\#g',A}\gets\mathcal T_{\#g',A}\cup \{T_{\mathrm{new}}\}$
                        \EndIf
                    \EndFor
                \EndFor
            \EndFor
            \State \Return $\mathcal T_{\#g, A}$
		\EndFunction
	\end{algorithmic}
\end{algorithm}

\begin{algorithm}[t]
\caption{STP-based circuit solver}
\label{alg:factor}
\begin{algorithmic}[1]
\Require Circuit topology $T$, Synthesis target $A$.
\Ensure Set of feasible gate decisions $\mathcal D$.

\Function {Solve}{$T, A$}
    \State Convert $T$ to canonical STP expressions $\Phi_{T,i}$
	\State Convert $A$ to structural matrices $M_i\in \mathbb{F}_2^{2\times 2^{n}}$
            %\State $D\gets \{\mathsf{CT, TC}\}^{\#g}$
    \State $\mathcal D\gets $\textsc{Factorize}($\{(\Phi_{T,i},M_i):i\in[n]\}$)
    \State \Return $\mathcal D$
\EndFunction
\vspace{0.6em}
\Function{Factorize}{$\{(\Phi_{T,i},M_i)\}$}
    \State Initialize the global decision set as $\mathcal D\gets\{\mathsf{CT},\mathsf{TC}\}^{\#g}$
    \State Initialize the equation pool as $\mathcal E\gets\{(\Phi_{T,i},M_i, \mathcal D)\}$ %\\\Comment{$D_0$ is the allowed decision set corresponding to the formula}
    \While{$\mathcal E\neq\varnothing$ and  $\mathcal D\neq\varnothing$}
        \State $(\Phi, M, D)\gets \mathcal E.\mathrm{pop}()$ with heuristic
        %\State $D\gets D-\mathcal D$ \Comment{Update $D$ and may cause backjumping}
        \State Suppose $\Phi=\Phi'\stp (I_{2^t}\ts N)$
        \If{$N=M_w$ or $M_r$}
            \State Factorize $M$ as $M=M'\stp (I_{2^t}\ts N)$ \Comment{Always feasible}
            \State $\mathcal E.\mathrm{push}(\Phi',M',D)$
        \ElsIf{$N$ corresponds to partial gate $g_j$}
            \For{$d\in\{\mathsf{CT},\mathsf{TC}\}$}
                \State $D_d\!\gets \!D$ with decision on $g_j$ restricted to $d$
                \If{$M=M'\stp (I_{2^t}\otimes M_{d})$ for some $M'$}
                    \State $\mathcal E.\mathrm{push}(\Phi',M',D_d)$
                \Else
                    \State remove $D_d$ from $\mathcal D$
                \EndIf
        \EndFor
        \EndIf
    \EndWhile
    \State \Return $\mathcal D$
\EndFunction
\end{algorithmic}
\end{algorithm}

\subsection{Factorization}\label{subsec:factorization_and_solver}

At the core of our circuit solver is the matrix factorization module, which is illustrated in more detail in Alg.~\ref{alg:factor}. For a fixed equation $\Phi_{T,i}=M_i$, the factorizer takes the rightmost factor of $\Phi_{T,i}$, say $I_{2^t}\ts N$, and attempts to factorize $M_i$ as $M_i'\stp(I_{2^t}\ts N)$, splitting that factor apart from the equation and then removing it. We refer to this operation as one \emph{factorization step}. When $N=M_w$ or $M_r$, such factorization is always feasible; when $N=M_{g_j}^{(\ell)/(h)}$ corresponds to an undetermined partial gate, the factorizer considers the two possible directions, $\mathsf{CT}$ and $\mathsf{TC}$, and attempts the factorization under each assignment. If an attempt fails, the corresponding gate decision is eliminated.

Rather than processing the equations in a fixed order, the factorizer maintains them in a pool and dynamically selects which equation to process next. For each equation, it estimates the factorization cost required to reach the next opportunity for reducing a gate decision, and prioritizes the equation with the lowest estimated cost. The selected equation is popped from the pool and advanced by one factorization step. By selecting equations smartly with heuristics, the factorizer eliminates infeasible gate decisions faster and terminates the overall solving process earlier, as shown in Fig.~\ref{fig:gate_dci_reduction}.

After each factorization step, the corresponding gate decisions are updated, and the residual equation is put back to the pool if further factors remain. If $N=M_w$ or $M_r$, no gate decision is affected. If $N=M_{g_j}^{(\ell)/(h)}$, the feasible directions of $g_j$ are restricted according to the factorization result; when both $\mathsf{CT}$ and $\mathsf{TC}$ remain feasible, the constraint may branch into two cases corresponding to the two assignments. The process then continues, until all equality constraints in the equation pool have been exhausted, indicating a clearance of constraints.

Conceptually, STP factorization serves as the local inference primitive for testing and eliminating gate decisions, similar to theory reasoning in SMT solving. At the global level, the solver explores gate decisions through branching and recursive elimination, resembling the classical DPLL algorithm used in solving SAT problems~\cite{dpll1,dpll2}.

More specifically, we next describe the factorization test used in each step. We use the following criterion to determine whether a rightmost STP factor can be extracted:
\begin{theorem}[STP Factorization Criterion]\label{thm:factorization}
    Suppose $M=\big[M_1 \, M_2 \, \cdots M_L\big]$ where $M_i$ are block matrices of the same size, and suppose $R=\big[\v e_{i_1} \,\v e_{i_2}\, \cdots\, \v e_{i_L}\big]$ is a rightmost factor where $\v e_i$s are the standard orthonormal basis. Then the following two conditions are equivalent:
    \begin{enumerate}
        \item Factorization $M=M'\stp R$ is feasible for some $M'$,
        \item $i_k=i_{k'}$ always implies $M_{i_k}=M_{i_{k'}}$ for any $k,k'\in[L]$.
    \end{enumerate}
\end{theorem}

% Next we show two concrete examples of the factorization process.
% \begin{example}
%     If
%     \begin{equation}
%         \setlength{\arraycolsep}{3pt} % 缩小列间距
%         M=M'\stp M_r=\begin{bmatrix}
%         M_a & M_b \\
%         \overline{M}_a & \overline{M}_b \\
%     \end{bmatrix},
%     \end{equation} then $M'$ can be factorized as
%     \begin{equation}
%         M'=\begin{bmatrix}
%         M_a & * & * & M_b \\
%         \overline{M}_a & * & * & \overline{M}_b \\
%     \end{bmatrix},
%     \end{equation}
%     where $*$ denotes variable matrices with the same size as $M_a, M_b$ that can take arbitrary values, and $\overline{M}_i$ denotes the entrywise negation of $M_i$.
% \end{example}

\begin{example}
    Denote $M_{\oplus}=\left[\begin{smallmatrix}
        1 & 0 & 0 & 1 \\
        0 & 1 & 1 & 0 \\
    \end{smallmatrix}\right]$ and suppose $M=\big[M_1\, M_2\,\cdots\, M_8\big]$. Then the factorization
    \begin{equation}
        M=M'\stp (I_2\ts M_\oplus)=M'\stp \left[\begin{smallmatrix}
            1 & 0 & 0 & 1 &  &  &  &  \\
            0 & 1 & 1 & 0 &  &  &  &  \\
             &  &  &  & 1 & 0 & 0 & 1 \\
             &  &  &  & 0 & 1 & 1 & 0 \\
        \end{smallmatrix}\right]
    \end{equation}
    is feasible if and only if:
    \begin{align}
        M_1=M_4,\ M_2=M_3,\\
        M_5=M_8,\ M_6=M_7.
    \end{align}
\end{example}

Furthermore, we identify a structural property specific to the parity functions induced by CNOT circuits that considerably simplifies STP factorization. For structural matrix of a parity Boolean function, a factorization with tail factor $I_{2^t}\otimes M_{\oplus}$ can be verified by inspecting only the first $1/2^t$ portion of the matrix. This is substantially cheaper than applying the generic factorization criterion (Theorem~\ref{thm:factorization}), which is used in original STP methods~\cite{pan2023Exact,pan2023Semitensor} and needs to inspect the entire structural matrix. Since factors of the form $I_{2^t}\otimes(\cdot)$ appear repeatedly in the canonical STP expressions of CNOT circuits, the simplification directly reduces the computational cost of our circuit solver.

\begin{theorem}
\label{theorem2}
    Suppose $\Phi(\v x)=p_{\v a}(\v x)=\bigoplus_{i=1}^n a_ix_i$ is a parity Boolean function, and $M_{p_{\v a}}=\big[M_1\,M_2\,\cdots\,M_{4\cdot 2^t}\big]$ is the structural matrix corresponding to $p_{\v a}$. Then factorization $M_{p_{\v a}}=M'\stp (I_{2^t}\otimes M_{\oplus})$ is feasible for $t\le n-2$ if and only if $M_1=M_4$ and $M_2=M_3$.
\end{theorem}

Due to space limitations, the detailed proof of this theorem is postponed to Appendix A.

\begin{figure}[t]
    \centering
    \includegraphics[width=0.49\textwidth]{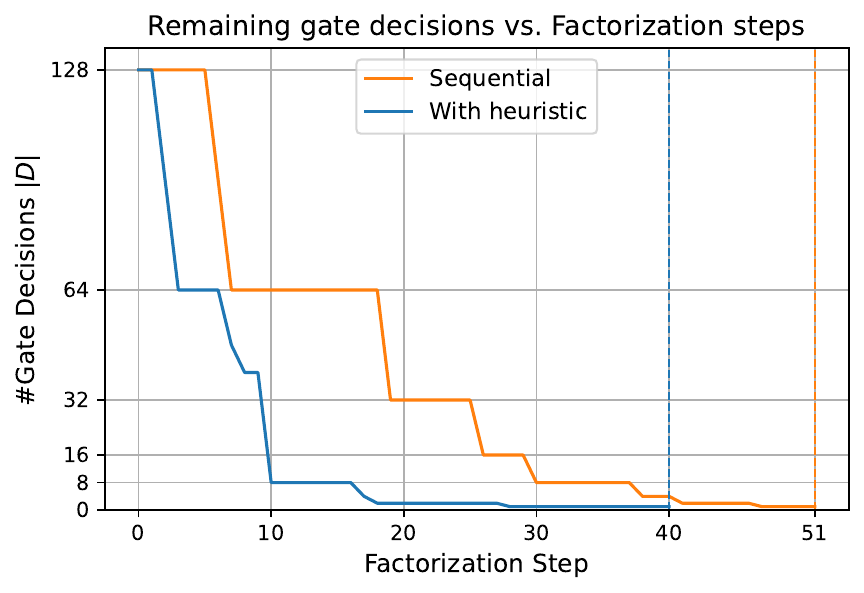}
    \caption{Number of remaining feasible gate decisions as factorization step increases, with different equality constraint selection strategies. The figure is plotted using a 7-gate topology that realizes a 5-qubit synthesis target. The behavior of the factorizer on other topologies and synthesis targets exhibits similar trends.}
    \label{fig:gate_dci_reduction}
\end{figure}

\subsection{Topology-Pruning Techniques}\label{subsec:topo-prune}
For $n$-qubit topologies with $g$ partial gates, the number of different topologies grows exponentially. For example, for $5$ qubits, there are roughly $\binom{5}{2}^6=1,000,000$ different $6$-gate topologies. Exhaustively enumerating this space is expensive even before circuit solving. Efficient exact synthesis therefore requires pruning infeasible topologies as early as possible. We next introduce the pruning strategies used in our algorithm.
% Throughout the section, we denote the input qubit variables as a vector
% \begin{equation}
%     \v x=[x_1,x_2,\cdots,x_n]^\top
% \end{equation}
% and
% \begin{equation}
%     \v x'=[x_1',x_2',\cdots,x_n']^\top=A\v x
% \end{equation}
% as the vector of desired output qubits.

\paragraph{Prune by variable mismatch} We can prune topologies by counting output qubit mismatches between the current topology and the synthesis target. Consider, for example, implementing $A=\left[\begin{smallmatrix}
    1 & 0 & 0 \\
    1 & 1 & 0 \\
    0 & 1 & 1 \\
\end{smallmatrix}\right]$ with the following topology $T$:
\begin{figure}[H]
\centering
\vspace{-.8em}
\begin{tikzpicture}
  \begin{yquant}[register/minimum height=4mm, operator/separation=4mm]
    qubit {$x_{\protect\the\numexpr\idx+1}$} q[3];
    %[control style={draw,fill=white,shape=yquant-rectangle,radius=3pt}, operator style={draw,fill=white,shape=yquant-rectangle,radius=3pt}]
    %cnot q[1]|q[2];
    [control style={draw,fill=white,shape=yquant-rectangle,radius=3pt}, operator style={draw,fill=white,shape=yquant-rectangle,radius=3pt}]
    cnot q[0]|q[1];
    output {$\{x_1,x_2\}=\supp(T,1)$} q[0];
    output {$\{x_1,x_2\}=\supp(T,2)$} q[1];
    output {$\{x_3\}\quad\,\, =\supp(T,3)$} q[2];
  \end{yquant}
\end{tikzpicture}
\vspace{-1em}
\end{figure}
Then we see the first and second output qubits may depend on $x_1$ and $x_2$, whereas the third output qubit can only depend on $x_3$. We refer to these sets of potentially relevant input variables as the \emph{supports} of the corresponding output qubits of $T$.

The target transformation requires $x_3'=x_2\oplus x_3$, but $x_2\notin \supp(T,3)$. Therefore, regardless of how the partial gates are instantiated, at least one additional CNOT targeting the third qubit is required. We call such a discrepancy a \emph{variable mismatch}. More generally, if the number of mismatched output qubits exceeds the number of gates remaining under the prescribed gate bound $\#g$, the topology cannot be extended to a valid solution and can be pruned early.

Formally, we recursively define the support of a topology by
\begin{align}
    \supp(T_\epsilon,i)&:=\{x_i\} \\
    \supp(T: G_{ij},i\text{ or }j)&:=\supp(T,i)\cup\supp(T,j)\\
    \supp(T: G_{ij},k\notin\{i,j\})&:=\supp(T,k)
\end{align}
where $T_\epsilon$ is the empty topology and $G_{ij}$ is the partial gate acting on qubits $i$ and $j$. Because the control and target of a partial gate have not yet been determined, the union is assigned to both  qubits, yielding an over-approximation of the input variables that may influence an output qubit. A partial topology $T$ can then be pruned if
\begin{equation}
    \#[\text{var. mismatch}]=\big|\{i:\mathrm{var}(x'_i)\nsubseteq\supp(T,i)\}\big|>\#g-|T|,
\end{equation}
where $\mathrm{var}(\v x')$ denotes the input qubit variables contained in $\v x'$. On the right-hand side, $\#g$ denotes the target gate number and $|T|$ is the number of gates in the topology $T$, which together give the count of remaining gates. The rule is valid because each subsequent partial gate, and hence each instantiated CNOT gate, can fix at most one such variable mismatch.

\paragraph{Prune by information exchange lower bound} A second pruning rule follows from the amount of information that must cross a qubit bipartition. For a qubit partition $\{1,2,\cdots,n\}=Q_1\uplus Q_2$\footnote{"$\uplus$" means \textit{disjoint union}, which denotes $Q_1\cup Q_2$ and, in the meantime, indicates $Q_1\cap Q_2=\varnothing$.}, we derive a lower bound on the number of CNOT gates that must act between $Q_1$ and $Q_2$, and hence on the number of cross-partition partial gates required in a valid topology.

Consider, for example, the linear function
\begin{align}
    x_1'&=x_1,\,\,\, x_2'=x_2, \\
    x_3'&=x_3\oplus x_1\\
    x_4'&=x_4\oplus x_1\oplus x_2,
\end{align}
together with partition $Q_1=\{1,2\}$ and $Q_2=\{3,4\}$. The outputs in $Q_2$ require two linearly independent parities from $Q_1$, namely $x_1$ and $x_1\oplus x_2$. Since each cross-partition CNOT can transfer at most one independent parity, at least two CNOT gates are necessary to implement the function.

Formally, for any partition $Q_1\uplus Q_2=[n]$, let $A_{Q_1Q_2}$ denote the submatrix of $A$ whose rows and columns are indexed by $Q_1$ and $Q_2$, respectively:
\begin{equation}
    (A_{Q_1Q_2})_{ij}=A_{q_iq_j}, 1\le i\le |Q_1|,1\le j\le |Q_2|,
\end{equation}
where $q_i$ and $q_j$ are the $i$-th and $j$-th qubit indices in $Q_1$ and $Q_2$, respectively. Any CNOT circuit implementing $A$ must then satisfy
\begin{multline}
    \#[\text{partial gate between } Q_1 \text{ and } Q_2\text{ in } T] \\ \ge \rank A_{Q_1Q_2}+\rank A_{Q_2Q_1},\forall \, Q_1\uplus Q_2=[n].
\end{multline}
The two rank terms quantify the independent information that must be exchanged across the partition in the two directions. Therefore, if a topology $T$ contains fewer cross-partition partial gates than this lower bound for any bipartition $Q_1\uplus Q_2=[n]$, it cannot implement $A$ and may be safely pruned.

To demonstrate the effectiveness of our pruning strategy, consider a simple case, where we consider all 6-gate topologies on 5 qubits\footnote{A random 5-qubit linear function $A$ is used in the experiment to derive the concrete numbers. Pruning strategies on other linear functions give similar results.}. Without pruning, there are $\binom{5}{2}^6=1,000,000$ candidate topologies. Applying the variable mismatch and information exchange pruning strategies individually reduces this number to approximately $28{,}061$ and $13{,}116$, respectively. When both strategies are applied, only $4{,}123$ candidate topologies remain, corresponding to a $243\times$ reduction in the search space and substantially decreasing the cost of the subsequent solving process.

\begin{table}[tb]
    \centering
    \begin{tabular}{cccc}
    \toprule
    \begin{tabular}[c]{@{}l@{}}Prune by\\ Var. Mismatch\end{tabular} & \begin{tabular}[c]{@{}l@{}}Prune by\\ Info. Exchange\end{tabular} & \#Topologies & Remaining Ratio \\
    \midrule
    No             & No                & 1,000,000      & 100\%       \\
    Yes            & No                & 28,061        & 2.8\%     \\
    No             & Yes               & 13,116         & 1.3\%      \\
    Yes            & Yes               & 4,123          & 0.4\%      \\
    \bottomrule
    \end{tabular}
    \caption{Number of 6-gate topologies that need to be solved on a 5-qubit random case} %Topologies are dramatically reduced using our topology solving and pruning techniques.}
    \label{table:topo_prune}
\end{table}

\subsection{Analysis of Partial Gate Width}\label{subsec:multi-qubit_partial_gate}
So far, our synthesis framework has been formulated using two-qubit partial gates. As introduced in Sec.~\ref{para:pre_partial_gates_topo}, partial gates may more generally span $k\geq 2$ candidate qubits. We now revisit this design choice and ask how the partial gate width $k$ affects the complexity of topology-and-solver-based exact synthesis algorithms.

Next we give a preliminary theory on this issue by minimizing the overall synthesis search space required to fully determine a CNOT circuit. To isolate the effect of $k$, we assume that the solver eliminates candidate decisions at a fixed average rate. Under this assumption, a smaller space of overall decision space is naturally preferable and decreases complexity.

Within our framework, a circuit is specified in two stages: first by selecting a topology and then by instantiating its partial gates through the solver. For a fixed $n$-qubit synthesis target requiring $g$ gates, the total number of joint choices is therefore:
\begin{align}
    f_{n,g}(k)&=\#\text{topologies}\times \#\text{gate decisions}\\
    &=\binom{n}{k}^g\times [k(k-1)]^g\\
    &=\left[\binom{n-2}{k-2}\cdot n(n-1)\right]^g,
\end{align}
which minimizes at $k=2$ or $n$, and maximizes at $k\sim n/2+1$ for fixed $n$ and $g$. The suboptimal values at $2<k<n$ reflects the redundancy introduced by intermediate partial gate widths: the same concrete CNOT gate can be instantiated from $\binom{n-2}{k-2}$ different $k$-qubit partial gate locations, thereby inflating the overall synthesis search space.

Although $k=2$ and $k=n$ yield the same minimum decision-space size, they lead to different synthesis algorithm architectures. With $k=2$, the choice of interacting qubits is handled by topology enumeration, while the solver determines only their control--target direction. This decomposition exposes independent topologies that can be enumerated and solved in parallel. On the other hand, the $k=n$ choice produces a single topology and transfers all workload to the solver, and is more complicated to parallelize. We therefore adopt $k=2$ throughout this paper, and leave the $k=n$ architecture for future work.

\section{Extension to Phase Polynomial Circuits}\label{sec:ext_to_phasepoly}

\begin{algorithm}[t]
	\caption{Exact synthesis algorithm for phase polynomial circuits}
    \label{alg:exact_syn_phasepoly}
	\begin{algorithmic}[1]
		\Require A linear parity matrix $A\in\mathrm{GL}(n,2)$ and a phase polynomial $p(\v x)$.
		\Ensure Optimal phase polynomial circuit(s) implementing the sum-over-path representation $(p,A)$.
		\Function {ExactSynPhasePolynomialSTP}{$p$, $A$}
			\State $ \#g\gets 0$ \Comment{Target gate number}
            \State $S\gets \varnothing$ \Comment{Solution set}
            \State $\mathcal P\gets\{\text{All parity terms $p_{\v a}$ in } p\}$
            \While{$S= \varnothing$}
                \State \Comment{Pruning adapted for phase polynomials}
                \State $\mathcal T_{\#g,A}\gets\textsc{EnumerateTopology}(\#g, A)$
                \For{$ T\in \mathcal T_{\#g, A}$} \Comment{Parallelizable}
                    \State $\mathcal D\gets \textsc{Solve}(T,A)$
                    \State $S_{\rm cand}\gets \textsc{Instantiate}(T, \mathcal D)$
                    \For{$C\in S_{\rm cand}$}
                        \If{$C$ reaches all parities in $\mathcal P$}
                            \State Insert $R_z(\theta)$ gates to $C$
                            \State \Comment{Insert phase gates where corresponding parities appear}
                            \State $S\gets S\cup\{C\}$
                        \EndIf
                    \EndFor
                \EndFor
                \State $\#g\gets \#g +1$
            \EndWhile
            \State \Return $S$
		\EndFunction
	\end{algorithmic}
\end{algorithm}

We next extend our framework to phase polynomial circuits composed of CNOT and $R_z$ gates. Following Amy et al.~\cite{amy2018CNOTcomplexity}, such circuits can be synthesized through a skeleton-and-checkpoint approach: a CNOT skeleton is first synthesized so that every required parity appears on some wire during its execution; then the corresponding $R_z$ rotations are inserted where those parities are reached. Exact  synthesis of phase polynomial circuits can therefore be reduced to finding a minimal CNOT skeleton that satisfies both the final linear transformation and all intermediate parity checkpoint constraints.

As mentioned in Sec.~\ref{sec:pre_qc}, a phase polynomial target is specified by a pair $(p,A)$, where $A\in\GL(n,2)$ is the final parity matrix and
\begin{equation}
    p(\v x)=\sum_{\v a\in\Ftwo^n}\theta_{\v a}p_{\v a}(\v x),
    \quad
    p_{\v a}(\v x)=a_1x_1\oplus\cdots\oplus a_nx_n
\end{equation}
is the phase polynomial, where each $\v a$ with $\theta_{\v a}\neq0$ defines a \emph{parity checkpoint} $p_{\v a}$ that needs to be reached. At any prefix of a CNOT skeleton, every wire carries a parity $p_{\v b}(\v x)$, with $\v b$ given by the corresponding row of the parity matrix of the circuit prefix. A checkpoint $p_{\v a}$ is reached if, at some circuit prefix and some wire, the parity carried by the wire is equal to $p_{\v a}$.

The pseudocode of the extended algorithm is shown in Alg.~\ref{alg:exact_syn_phasepoly}. It generates CNOT skeletons satisfying the final constraint $A$ and retains only those that reach all required parity checkpoints. For every surviving skeleton, the $R_z$ rotation specified by $\theta_{\v a}$ is inserted where checkpoint $\v a$ is reached.

We also adapt the topology pruning strategies to account for the checkpoint constraints. For each partial topology, we estimate the minimum number of additional gates required before a given checkpoint parity can be reached. Let $T$ be a partial topology, and for a required parity $p_{\v a}$, let $\supp(p_{\v a}):=\{i:a_i=1\}$. The minimum number of additional gates needed to make $p_{\v a}$ available on wire $q$ can be lower bounded by the following set-cover bound:
\begin{multline}
    \mathrm{cover}_q(T, p_{\v a})=\min\Big\{\,|Q|:\;Q\subseteq[n]-\{q\},\\
        \supp(p_{\v a})
        \subseteq\bigcup_{q'\in Q\cup\{q\}}\supp(T,q')\,\Big\}.
\end{multline}
Since a checkpoint may be reached on any wire, we take its minimum over all possible wires $q$ ($\min_{q\in [n]}$). Then we can prune a topology if the number of remaining gates is not enough to reach any required parity checkpoint ($\max_{p_{\v a}\in p}$):
\begin{equation}
    \max_{p_{\v a}\in p} \min_{q\in{[n]}} \mathrm{cover}_q(T,p_{\v a})>\#g -|T|.
\end{equation}
The original variable mismatch and information exchange lower bound pruning techniques are also retained alongside this checkpoint-specific rule.

%A tighter integration for parity checkpoint constraints would encode each checkpoint as an existential disjunction over possible places that it may occur in the circuit inside the circuit solver; we leave this alternate approach for future work.

\section{Experiments}
\label{sec:experiments}
In this section, we test our STP-based method on both random and real-world circuits to examine its performance, and compare it against existing exact synthesis approaches. Our algorithm is implemented in Rust\footnote{The source code is available upon reasonable request.}, and is parallelized via the \texttt{Rayon} package.   All experiments were run on a Linux server with Intel(R) Xeon(R) Platinum 8253 CPUs at 2.20 GHz and 3 TB of memory.

We use SAT-based exact synthesis as the primary baseline, as it is more scalable and flexible than database-based methods~\cite{CNOT_opt_UCL,CNOT_opt_group_theoretic} for the problem settings considered here. Specifically, we use Qiskit-SAT~\cite{qiskit-sat}, a Qiskit plugin for exact Clifford and CNOT synthesis based on the Z3~\cite{Z3} solver, as the representative baseline from the quantum computing community. For parallel experiments, we additionally replace its backend with the multi-threaded Gimsatul solver~\cite{fleury2022gimsatul}, an open-source parallel SAT solver that is easy to integrate and has demonstrated performance competitive with the winner of the SAT Competition 2021 parallel track. Unless otherwise stated, we will use ``\qiskitsat{}'' to refer to its original configuration with Z3 backend, and use ``Gimsatul'' to refer to the corresponding configuration using the multi-threaded Gimsatul backend.

%Existing exact CNOT-synthesis techniques can be broadly divided into database-based methods~\cite{CNOT_opt_UCL,CNOT_opt_group_theoretic} and solver-based methods~\cite{shaik2024Optimal,qiskit-sat}. Database-based methods are typically limited to small qubit counts because of their substantial memory requirements and often exploit equivalence under qubit permutations to reduce the size of the circuit library, whereas our method does not take permutation equivalence into consideration. Meanwhile, the SAT-based method has stood out as the most effective and flexible method among the solver-based methods, and has been widely used in classical exact synthesis. We therefore focus on SAT-based synthesis as the primary baseline. Specifically, we take the Qiskit implementation of the SAT-based method~\cite{qiskit-sat}, denoted as \qiskitsat{}, as a representative SAT-based exact-synthesis implementation developed for the quantum-computing research community. For experiments involving parallelism, we additionally use Gimsatul~\cite{fleury2022gimsatul} as a representative parallel SAT solver and integrate it as the backend of \qiskitsat{}. Unless stated otherwise, \qiskitsat{} refers to the original configuration with its default single-threaded Z3 backend, whereas ``Gimsatul'' refers to the corresponding configuration using the multithreaded Gimsatul backend.

\subsection{Experiment: Random CNOT Circuits}\label{sec:exp_rand_circ}

\begin{figure}
    \centering
    \includegraphics[width=0.99\linewidth]{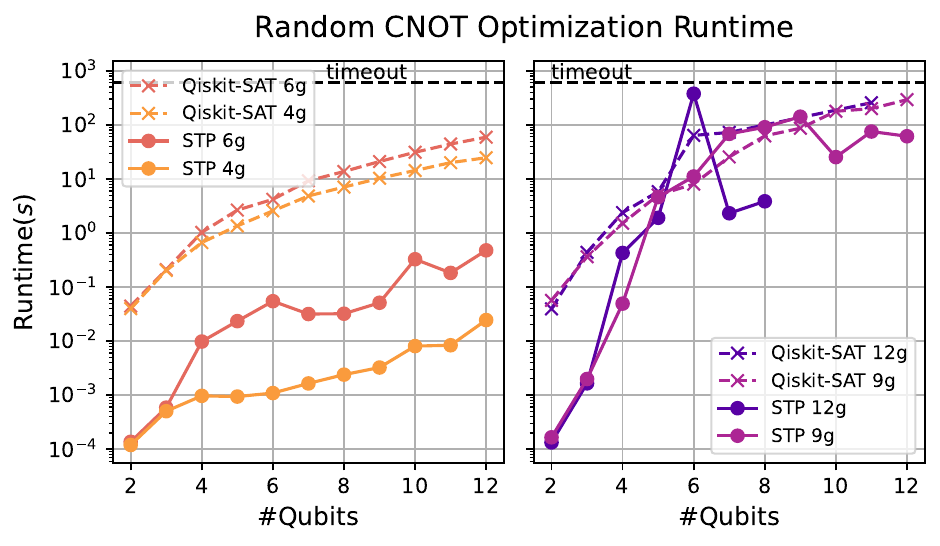}

    \caption{Runtime of STP against \qiskitsat{} on random CNOT circuits. For each data point, 10 random circuits are generated and tested. The timeout value is set to 600 seconds.}
    \label{fig:rand_circ_scalability}
\end{figure}

We first evaluate our algorithm on randomly generated linear functions. For each qubit number $n$ and gate count $g$, we generate 10 random CNOT circuits and use their resulting parity matrices $A_{\rm rand}$ as synthesis targets. The larger $g$, the more complicated $A_{\rm rand}$ is. We compare STP with \qiskitsat{} over a range of $(n,g)$ configurations, with the results shown in Fig.~\ref{fig:rand_circ_scalability}. For each configuration, we report the average runtime over the 10 instances, with a timeout of 600 seconds per instance.

As shown in Fig.~\ref{fig:rand_circ_scalability}, STP substantially outperforms \qiskitsat{} for configurations with $g\leq 6$ and remains competitive on more complicated instances with $g\approx 9$. When the number of qubits is small and the target linear function is moderately difficult, STP is typically $2$--$3$ orders of magnitude faster than \qiskitsat{}, corresponding to a speedup of approximately $100$--$1000\times$.

This advantage primarily comes from the lower overhead of our circuit-specific formulation. Even for small instances, a SAT-based approach typically introduce hundreds of Boolean variables and thousands of clauses to encode the synthesis problem. In contrast, the STP formulation is more closely tailored to the algebraic structure of CNOT circuits and therefore admits a more compact representation. Moreover, SAT-based methods invoke an external general-purpose SAT solver, whereas STP integrates the search procedure natively and avoids additional solver-invocation overhead.

As $g$ increases and the target linear functions become more difficult, the performance gap gradually narrows. In this regime, encoding and invocation overhead account for a smaller fraction of the total runtime, while the underlying combinatorial search gradually dominates. Mature SAT solvers benefit from decades of innovation and optimization, and have therefore been highly effective on such difficult instances. %This trend is also consistent with the NP-hard nature of exact CNOT-circuit synthesis: a circuit-specific formulation can substantially reduce practical overhead, but does not eliminate the worst-case combinatorial complexity. We therefore envision STP as a fast circuit-resynthesis kernel for small and moderately difficult subproblems within practical quantum-circuit optimization workflows, complementing SAT-based methods that are better suited to more challenging instances.
Thus, the main advantage of STP lies in efficiently solving small and moderately difficult synthesis instances, making it well suited as a fast local re-synthesis kernel that complements SAT-based methods on harder cases.

\begin{figure}
    \centering
    \includegraphics[width=0.99\linewidth]{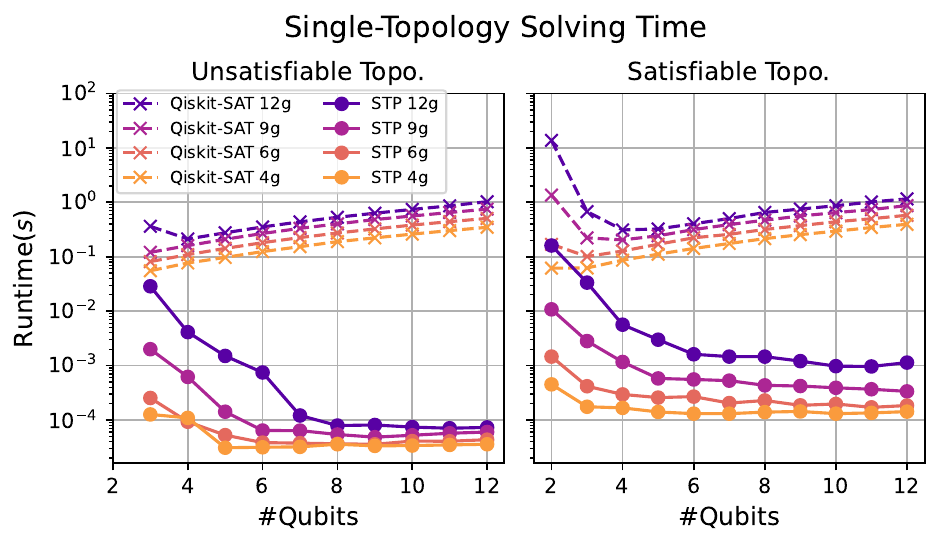}

    \caption{Time required to solve random single topologies by STP circuit solver and Qiskit-SAT. \textbf{Left}: time required for unsatisfiable topologies, for which the sampled topology cannot realize the synthesis target; \textbf{Right}: time required for satisfiable topologies, for which at least one gate direction realization exists.}
    \label{fig:single_topo_solving}
\end{figure}

\subsection{Experiment: Single Topology Solving}\label{sec:exp_single_topo}
To better understand the source of STP’s performance advantage over Qiskit-SAT, we further evaluate the time required by each method to solve a \emph{single topology}. For each qubit–-gate configuration $(n,g)$, we randomly generate topology--target pairs and measure the solving times of the STP circuit solver and Qiskit-SAT. We consider two classes of instances: \emph{unsatisfiable topologies}, for which the sampled topology cannot realize the synthesis target, and \emph{satisfiable topologies}, for which at least one realization exists. For each configuration and satisfiability class, the solving time is averaged over 10 randomly generated cases. The results are presented in Fig.~\ref{fig:single_topo_solving}.

As shown in Fig.~\ref{fig:single_topo_solving}, the STP circuit solver consistently outperforms Qiskit-SAT. It achieves speedups above $500\times$ in most configurations, with a maximum speedup of $1.4\times 10^4$. Although satisfiable topologies generally take longer for STP to solve than unsatisfiable ones, STP maintains a substantial advantage over Qiskit-SAT in both classes. This efficiency the single topology level lays the foundation for the end-to-end performance advantage of STP-based exact synthesis reported in Sec.~\ref{sec:exp_rand_circ}.

Interestingly, the solving time initially decreases as the number of qubits increases. This is because, with the gate count fixed, increasing the circuit width makes the sampled topologies sparser and reduces interactions among partial gates, thereby lowering the combinatorial complexity. As the number of qubits grows further, however, the runtime of Qiskit-SAT begins to increase steadily: additional qubits introduce more Boolean variables and constraints, increasing both encoding and solving overhead. On the contrary, STP exhibits much weaker dependence on circuit width. Although each topology is also encoded as STP expressions in the STP circuit solver, the encoding is more closely tailored to the algebraic structure of the circuit, and its overhead barely grows as the number of qubits increases.

\begin{figure}
    \centering
    \includegraphics[width=0.99\linewidth]{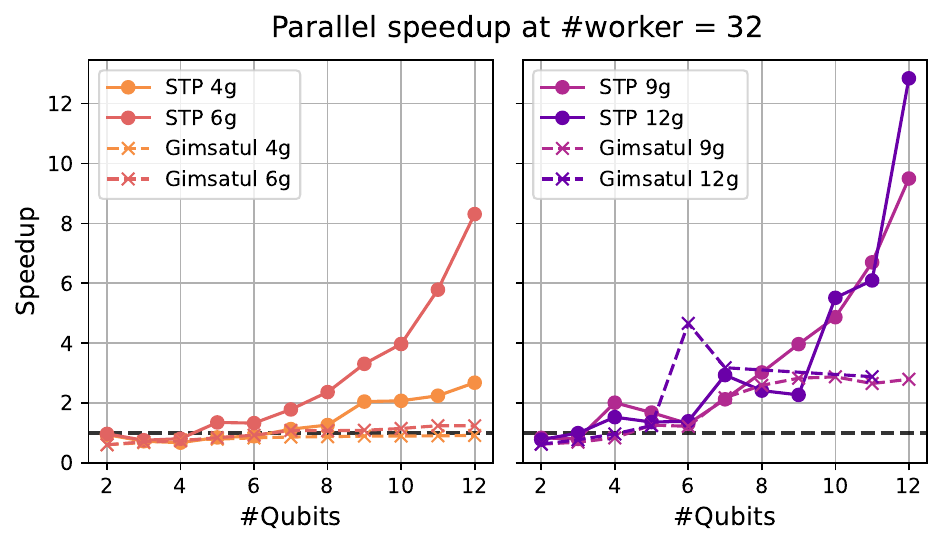}

    \caption{Parallel speedup of STP against Qiskit-SAT with the Gimsatul backend~\cite{fleury2022gimsatul} with 32 workers on random CNOT circuits. For each data point, 10 random circuits are generated and tested.}
    \label{fig:rand_circ_paral_speedup}
\end{figure}

\subsection{Experiment: Parallel Speedup}\label{sec:exp_paral_speedup}

%Modern processors increasingly rely on multicore architectures, making effective use of the available cores essential to translating hardware parallelism into wall-clock speedup. Parallel SAT solving has received substantial attention~\cite{martins2012overview}; nevertheless, obtaining robust and scalable speedups remains highly challenging~\cite{katsirelos2013resolution}. In contrast, our method exposes parallelizable structure at the topology level: the enumeration branches leading to distinct topologies, and the subsequent topology-solving tasks, can be assigned to workers independently. This structure is designed to yield more stable scaling as additional workers are introduced.

We next evaluate the parallel speedup of our method against the SAT-based baseline. For the SAT baseline, we use Qiskit-SAT with the multi-threaded Gimsatul backend, denoted as ``Gimsatul.'' Both methods are evaluated on the random CNOT-circuit instances described in Sec.~\ref{sec:exp_rand_circ}, using different numbers of workers. For an instance executed with $p$ workers, the parallel speedup is defined as $S_p=T_1/T_p$,
where $T_1$ and $T_p$ denote the wall-clock runtimes using one and $p$ workers, respectively. The reported speedups are averaged over the random instances in each $(n,g)$ configuration.

The results for 32 workers are shown in Fig.~\ref{fig:rand_circ_paral_speedup}; results for smaller worker counts exhibit similar trends and are provided in the supplementary material. A clear separation can be observed between the parallel scalability of STP and Gimsatul. For relatively small and easy instances (left figure), the speedup of Gimsatul generally remains close to $1$, indicating limited benefit from additional threads. Its average speedup improves on more difficult instances (right figure) but remains below $5$ across all tested configurations.

In contrast, the parallel speedup of STP increases consistently with problem difficulty. For deeper and more complicated instances, STP consistently achieves speedups above $5\times$, with a maximum observed speedup of $12.8\times$ on 32 workers at the most difficult configuration with $(n,g)=(12,12)$. For both methods, speedups below $1$ may occur on small instances because the parallelization overhead exceeds the computation time saved via parallelism.

This trend is consistent with the topology-level parallel structure of STP. For small instances, the workload is insufficient to amortize task-distribution and synchronization overhead. As the problem size increases, more independent topologies can be processed concurrently, allowing STP to utilize the available workers more effectively and thereby resulting in higher parallel speedup.

\subsection{Experiment: Circuit Peephole Rewriting Optimization on Real-World Circuits} \label{sec:cnot-resynthesis}

% Generated by exp/anal/build_phase_poly_table_v7.py.
\begin{table*}[tb]
\caption{Circuit peephole rewriting optimization on QASMBench circuits with $1000< \#\mathrm{gate}\le 10000$.}
\label{table:rawbig_qasmbench}
\begin{adjustbox}{max width=\textwidth, center}
\begin{tabular}{lccccc|ccccc}
\toprule
\multirow{2}{*}[-0.8ex]{Case} &
\multirow{2}{*}[-0.8ex]{\#qubits} &
\multirow{2}{*}[-0.8ex]{\#gates} &
\multirow{2}{*}[-0.8ex]{\makecell{\#CNOT\\shards}} &
\multirow{2}{*}[-0.8ex]{\makecell{\#Phase Poly.\\shards}} &
\multirow{2}{*}[-0.8ex]{\#CNOT} &
\multirow{2}{*}[-0.8ex]{\makecell{Qiskit-SAT~\cite{qiskit-sat}\\Runtime($s$)}} &
\multirow{2}{*}[-0.8ex]{\makecell{Gimsatul~\cite{fleury2022gimsatul}\\Runtime($s$) (=1)}} &
\multirow{2}{*}[-0.8ex]{\makecell{\textbf{STP (ours)}\\Runtime($s$)}} &
\multicolumn{2}{c}{\textbf{STP+Gimsatul Portfolio}} \\
\cmidrule(lr){10-11}
& & & & & & & & & Runtime($s$) & STP wins \\
\midrule
ising\_n98 & 98 & 1072 & 48 & 49 & 194 & 11.9 (0.4$\times$) & 5.3 & \textbf{1.5 (3.6$\times$)} & 1.7 (3.0$\times$) & 48/48 (100\%) \\
swap\_test\_n115 & 115 & 1085 & 57 & 57 & 456 & 32.0 (0.4$\times$) & 11.5 & \textbf{2.5 (4.7$\times$)} & 2.7 (4.2$\times$) & 57/57 (100\%) \\
knn\_129 & 129 & 1218 & 64 & 64 & 512 & 34.2 (0.4$\times$) & 13.0 & 3.6 (3.6$\times$) & \textbf{3.0 (4.3$\times$)} & 64/64 (100\%) \\
qugan\_n71 & 71 & 1415 & 70 & 103 & 552 & 21.9 (0.4$\times$) & 7.9 & \textbf{2.6 (3.1$\times$)} & 2.8 (2.8$\times$) & 35/35 (100\%) \\
adder\_n118 & 118 & 1834 & 111 & 153 & 845 & 1422 (0.2$\times$) & 245 & 56.1 (4.4$\times$) & \textbf{26.3 (9.3$\times$)} & 158/184 (86\%) \\
dnn\_n8 & 8 & 2032 & 0 & 60 & 192 & 0.61 (0.9$\times$) & 0.53 & 0.53 (1.0$\times$) & \textbf{0.53 (1.0$\times$)} & 0/0 (NA) \\
qft\_n29 & 29 & 2059 & 323 & 100 & 812 & 402 (0.2$\times$) & \textbf{83.8} & \textgreater 3600 & 87.7 (1.0$\times$) & 19/95 (20\%) \\
qugan\_n111 & 111 & 2235 & 110 & 163 & 872 & 28.6 (0.5$\times$) & 15.4 & 7.1 (2.2$\times$) & \textbf{7.0 (2.2$\times$)} & 55/55 (100\%) \\
wstate\_n380 & 380 & 2275 & 76 & 0 & 758 & 725 (0.1$\times$) & 41.1 & 3.1 (13.1$\times$) & \textbf{2.9 (14.1$\times$)} & 140/140 (100\%) \\
vqe\_uccsd\_n6 & 6 & 2282 & 274 & 150 & 1052 & 525 (0.1$\times$) & 58.7 & \textbf{12.0 (4.9$\times$)} & 13.0 (4.5$\times$) & 312/312 (100\%) \\
basis\_trotter\_n4 & 4 & 2490 & 56 & 150 & 582 & 27.2 (0.3$\times$) & 7.3 & 49.1 (0.1$\times$) & \textbf{4.8 (1.5$\times$)} & 57/58 (98\%) \\
gcm\_h6 & 13 & 3148 & 92 & 247 & 762 & 32.6 (0.6$\times$) & 21.0 & \textbf{13.8 (1.5$\times$)} & 14.2 (1.5$\times$) & 82/82 (100\%) \\
knn\_341 & 341 & 3232 & 170 & 170 & 1360 & 98.2 (0.4$\times$) & 40.9 & 16.8 (2.4$\times$) & \textbf{15.3 (2.7$\times$)} & 170/170 (100\%) \\
swap\_test\_n361 & 361 & 3422 & 180 & 180 & 1440 & 105 (0.4$\times$) & 44.3 & \textbf{18.6 (2.4$\times$)} & 18.6 (2.4$\times$) & 180/180 (100\%) \\
dnn\_n16 & 16 & 4064 & 0 & 124 & 384 & 9.4 (0.4$\times$) & 3.9 & 3.9 (1.0$\times$) & \textbf{3.8 (1.0$\times$)} & 0/0 (NA) \\
ising\_n420 & 420 & 4614 & 209 & 210 & 838 & 46.2 (0.7$\times$) & 31.5 & \textbf{15.8 (2.0$\times$)} & 17.1 (1.8$\times$) & 209/209 (100\%) \\
multiplier\_n45 & 45 & 5981 & 429 & 453 & 2574 & 332 (0.5$\times$) & 160 & 125 (1.3$\times$) & \textbf{103 (1.6$\times$)} & 600/605 (99\%) \\
adder\_n433 & 433 & 6769 & 408 & 573 & 3120 & 5125 (0.2$\times$) & 1144 & 340 (3.4$\times$) & \textbf{233 (4.9$\times$)} & 587/683 (86\%) \\
qugan\_n395 & 395 & 8057 & 394 & 589 & 3144 & 310 (0.8$\times$) & 243 & 217 (1.1$\times$) & \textbf{216 (1.1$\times$)} & 197/197 (100\%) \\
qft\_n63 & 63 & 9828 & 901 & 270 & 3906 & 2998 (0.1$\times$) & \textbf{317} & \textgreater 3600 & 335 (0.9$\times$) & 24/265 (9\%) \\
\midrule
\textbf{Median speedup} &  &  &  &  &  & 0.4$\times$ & 1.0$\times$ & 2.4$\times$ & 2.3$\times$ &  \\
\bottomrule
\end{tabular}
\end{adjustbox}
\end{table*}

\begin{figure*}[ht]
    \centering
    \includegraphics[width=0.9\textwidth]{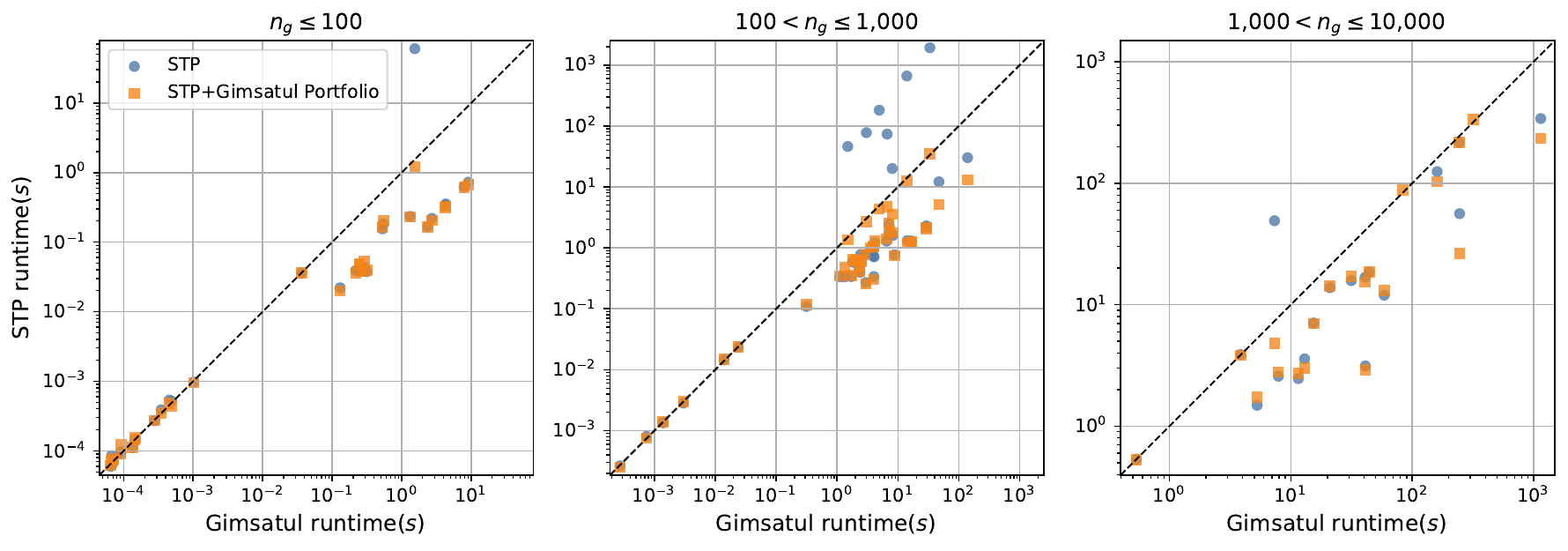}
    \caption{Optimization runtime of Gimsatul, STP and STP+Gimsatul portfolio on QASMBench~\cite{li2023qasmbench}. \textbf{Left}: Cases with gate count no more than 100; \textbf{Middle}: Cases with gate count greater than 100 but no more than 1000; \textbf{Right}: Cases with gate count greater than 1,000 but no more than 10,000.}
    \label{fig:runtime_qasmbench}
\end{figure*}

\begin{figure}[ht]
    \centering
    \includegraphics[width=0.45\textwidth]{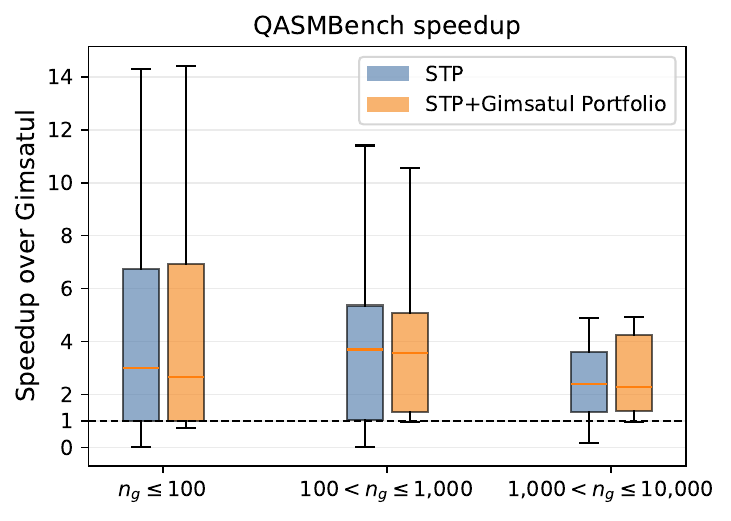}
    \caption{Speedup of STP-based kernels over Gimsatul on different categories of QASMBench}
    \label{fig:acc_qasmbench}
\end{figure}

  %<overall experiment setup>
  The random-circuit experiments isolate the performance of the exact-synthesis kernel. We next show how this advantage can be translated into a complete optimization workflow applicable to real-world quantum circuits. Following the use of exact synthesis as a local rewriting engine in classical logic synthesis~\cite{soeken2018Practical}, we integrate STP into a circuit re-synthesis pass for
  quantum circuits. Instead of attempting to find the optimal implementation of an entire CNOT circuit, which would be intractable for large real-world circuits, we repeatedly apply exact synthesis to appropriately sized local circuit regions, which is a form of peephole-optimization.
  %<while there may be another circuit optimization workflow that exhibits better overall gate complexity performance... this is just one possible way to show the advantage that a better exact synthesis kernel may bring us.>

\paragraph{Benchmark}
  We evaluate different algorithms on QASMBench~\cite{li2023qasmbench}, a
quantum circuit benchmark suite designed for evaluating NISQ systems and software.
It contains circuits from a broad range of application domains, including
state preparation, quantum chemistry, cryptography, quantum machine
 learning, optimization, and etc. QASMBench
has also been adopted by a number of recent studies on quantum circuits and architectures~\cite{wang2024atomique,tan2024compiling,tao2025quantum}.

We use QASMBench test cases that contain at most 10,000 gates. This gives 96 circuits,
which we divide by gate count $n_g$ into three disjoint categories:
32 circuits with $n_g\leq 100$, 44 with
$100<n_g\leq 1{,}000$, and 20 with
$1{,}000<n_g\leq 10{,}000$. The last group contains circuits with up to
433 qubits and 9,828 gates.

\paragraph{Optimization workflow and algorithm}
  For each circuit, every maximal sequence of consecutive CNOT+$R_z$ gates is extracted as a phase polynomial subcircuit, while all other gates remain unchanged.
  A subcircuit may still be too difficult to be optimized as a whole. We therefore further partition it from left to right into smaller circuit \emph{shards} before they are optimized by exact synthesis kernels. More specifically, the circuit shards are partitioned as follows. For each prefix of a subcircuit, we compute its parity
  matrix $A$ and a target-dependent complexity score $H_{\rm CNOT}(A)$, which is a heuristic that predicts the minimal number of CNOT gates required to implement $A$. Then, the longest prefix satisfying $H_{\rm CNOT}(A)\leq H_0$ is selected as the next shard.
  The threshold $H_0$ is set to 6, so it is beyond the typical ability of database-based methods~\cite{CNOT_opt_UCL}, and not too big so it becomes infeasible for exact synthesis kernels.

  The circuit shards are then sent to different exact synthesis kernels (algorithms) for optimization. The kernels we evaluate include:
  \begin{enumerate}
    \item Qiskit-SAT with the single-threaded Z3 backend;
    \item Gimsatul, implemented as Qiskit-SAT with a 32-thread Gimsatul backend;
    \item STP, our algorithm running with 32 workers; and
    \item STP+Gimsatul Portfolio, a portfolio that runs a 16-worker STP instance and a 16-thread Gimsatul instance concurrently and returns as soon as either solver completes.
  \end{enumerate}
  For the STP-based kernel, shards consisting of only CNOT gates are synthesized using its CNOT mode (Alg.~\ref{alg:exact_syn}), and shards consisting of both CNOT and $R_z$ gates are synthesized using its phase polynomial mode (Alg.~\ref{alg:exact_syn_phasepoly}). Since \qiskitsat{} does not natively support phase-polynomial optimization, to compare both algorithms on an equal footing, we enhance it with the encoding method of Li et al.~\cite{li2025HOPPS}, enabling it to serve as an exact synthesis kernel for phase polynomial shards as well. After optimization, the synthesized shards are reassembled with the unchanged gates to produce the final circuit.

  Optimizing each shard independently may miss optimization opportunities across shard boundaries. For example, two identical neighboring CNOT gates may remain unchanged if they are assigned to different shards, but they can actually be directly canceled out. So after the initial optimization pass that sweeps the entire circuit, we subsequently perform two rounds\footnote{We use two rounds of IBO because experiments show that most cross-boundary optimization chances are found in the first two rounds.} of additional random shard optimization, which picks circuit shards across previous boundaries randomly and optimizes them using the corresponding exact synthesis kernel. This strategy is also used in prior work~\cite{li2025HOPPS}, and is called \emph{iterative blockwise
  optimization} (IBO). We would follow their terminology. Each IBO task selects a window crossing the boundary between two previously optimized shards. %At least 30\% of the selected gates must come from either side of the boundary to ensure, and the resulting target must satisfy the same $H(A)\leq H_0$ constraint.

  \paragraph{Results} The experimental results are summarized in Fig.~\ref{fig:runtime_qasmbench}, with detailed results for the circuits containing more than $1{,}000$ gates reported in Table~\ref{table:rawbig_qasmbench}. Throughout this evaluation, we use the 32-thread Gimsatul configuration as the baseline, thereby comparing STP and the SAT-based method under the same overall parallelism budget.

  Fig.~\ref{fig:runtime_qasmbench} compares the end-to-end runtimes of STP-based methods and Gimsatul. In all three circuit-size categories, most points lie below the diagonal, indicating that STP is faster. Nevertheless, a small number of outliers show that pure STP can be slower than the SAT-based method on certain instances. To improve algorithm robustness, we propose an STP+Gimsatul portfolio kernel that runs a 16-worker STP instance and a 16-thread Gimsatul instance concurrently and returns as soon as either solver completes. This configuration is labeled ``STP+Gimsatul Portfolio'' in the figures. By exploiting the complementary power of the two approaches, the portfolio eliminates all the runtime outliers exhibited by pure STP.

  Across the full benchmark suite, pure STP achieves a median speedup of $3.41\times$ over the 94 instances that it completes within the 3,600s overall timeout, whereas the portfolio achieves a median speedup of $3.20\times$ and completes all 96 instances. Pure STP outperforms Gimsatul on 73 of its 94 completed instances, while the portfolio is faster on 81 of the 96 instances and incurs no timeout. In contrast, the single-threaded Qiskit-SAT configuration is slower than Gimsatul on almost all instances, with an overall median speedup of only $0.35\times$, highlighting the importance of effective parallelization.

The speedup distributions in Fig.~\ref{fig:acc_qasmbench} further show that the advantage of STP is not confined to certain circuit size categories. For $n_g\leq 100$, the median speedups of pure STP and the portfolio are $2.99\times$ and $2.67\times$, respectively; and the maximum speedup achieves $14.3\times$ and $14.4\times$. The median speedups increase to $3.70\times$ and $3.57\times$ for $100<n_g\leq 1{,}000$, and remain at $2.41\times$ and $2.29\times$ for $1{,}000<n_g\leq 10{,}000$. Though the maximum speedup drops for larger circuits, the median speedup remains stable. Thus, STP-based kernels provide consistent acceleration across different circuit-size regimes. Although pure STP is slower than Gimsatul on a small number of instances, the STP+Gimsatul portfolio substantially improves performance stability by retaining most of the STP speedup while providing a SAT-based fallback so as to avoid slowdowns.

Table~\ref{table:rawbig_qasmbench} reports all 20 circuits in the
largest category in more detail. As different exact synthesis kernels give the same optimized circuits, we mainly compare the runtimes of different algorithms. The STP-based configuration achieves the fastest runtime on 18 of the 20 cases: pure STP is fastest on 7 cases and the portfolio on 11. Although the portfolio configuration slightly sacrifices speed on small cases, it reaches speed comparable to pure STP on large cases by harnessing the complementary power between STP and SAT solvers. More specifically, the last column of Table~\ref{table:rawbig_qasmbench} reports the number of nontrivial shard optimization runs won by STP in the STP+Gimsatul portfolio, showing that STP achieves a win rate above 80\% on the vast majority of cases.
The portfolio configuration also avoids performance regressions: The pure STP configuration exceeds the one-hour limit on \texttt{qft\_n29} and \texttt{qft\_n63}; the portfolio nevertheless completes them in 88\,s and 335\,s, close to the corresponding Gimsatul times of 83.8\,s and 317\,s. These results show that the portfolio retains most of the performance advantage of STP while substantially improving robustness on occasional instances that are unfavorable to the STP approach.

\section{Related Works}

%Quantum circuit synthesis has been studied from many angles, including exact and asymptotically optimal decompositions for general unitaries~\cite{csd,Gray_code_for_U_syn,qsd}, state preparation~\cite{iten2016quantum,plesch2011quantum}, and topology-aware numerical search~\cite{qsearch}. These methods target broader gate sets and problem classes, while this paper focuses on exact CNOT and parity-network kernels.

\paragraph{Classical exact synthesis and STP-based exact synthesis}
Exact synthesis has been extensively studied in classical logic synthesis, and readers can refer to~\cite{soeken2018Practical} for a comprehensive review. A typical workflow of SAT-based classical exact synthesis involves enumerating different topology families, and instantiating each partial gate using SAT solvers~\cite{haaswijk2020SATBased}. More recently, Pan and Chu~\cite{pan2023Exact,pan2023Semitensor} explored replacing SAT solving with STP-based circuit solving, using DAG circuit topologies together with STP factorization for exact synthesis and logic rewriting. Our work is inspired by this line of research, but we have redesigned the STP circuit solver to adapt to CNOT circuits and further introduced topology-pruning and factorization optimization techniques that exploit the specific algebraic structure of parity Boolean functions.

\paragraph{Optimization of CNOT and phase polynomial circuits} The synthesis of linear reversible circuits, or CNOT circuits, was first studied by Patel, Markov, and Hayes, who gave an asymptotically optimal synthesis algorithm for CNOT circuits~\cite{pmh_optimal_linear_synthesis}. More recent work explores further optimization of CNOT circuits using greedy algorithm~\cite{debrugiere2021Gaussian, CNOT_opt_UCL}, A* search~\cite{CNOT_opt_UCL} and database lookup~\cite{CNOT_opt_UCL,CNOT_opt_group_theoretic}. There has also been extensive study on optimization of phase polynomial circuits, including GraySynth~\cite{amy2018CNOTcomplexity} and PhasePoly~\cite{chen2025phasepoly} that uses A* search. Compared with these heuristic approaches, our method provides optimal synthesis and adopts a different technical route: topology enumeration combined with a domain-specific STP circuit solver.

\paragraph{Exact quantum circuit synthesis}
%Exact synthesis of general quantum circuits is $\mathsf{QMA}$-hard~\cite{janzing2005Nonidentitycheck};
Amy, Azimzadeh and Mosca showed that the exact synthesis of CNOT+$R_z$ is $\mathsf{NP}$-hard under certain conditions in 2018~\cite{amy2018CNOTcomplexity}, which inspired later works to use SAT solvers to tackle the problem. Shaik and Pol~\cite{shaik2024Optimal} encoded the problem into SAT, classical planning and quantified Boolean formula instances and used to existing solvers to solve it, and found that
the SAT-based method is the most effective one among all three approaches. Ivrii and Treinish~\cite{qiskit-sat} developed a Qiskit plugin based on the Z3 solver~\cite{Z3} for the exact synthesis of Clifford and CNOT gates, which serves as the off-the-shelf method for quantum exact synthesis, and is the primary baseline used in our experiments. HOPPS~\cite{li2025HOPPS} encodes the exact synthesis of phase polynomial circuits into SAT, and has already been used in our experiment. There are also attempts that use precomputed tables to find the optimal implementation of CNOT circuits~\cite{CNOT_opt_UCL}.

\section{Conclusion}

We presented an STP-based exact synthesis framework for CNOT and phase polynomial circuits. The method combines topology enumeration and pruning with an STP-based circuit solver that instantiates partial gates through matrix factorization. Compared with SAT-based and database-based approaches, the framework introduces lower encoding overhead, requires no precomputed circuit database, and can naturally solve different topologies in parallel.

Experiments on random synthesis targets and QASMBench circuits show that STP is particularly effective for small and moderately difficult exact-synthesis tasks and can accelerate practical peephole rewriting workflows. We also showed that an STP+SAT portfolio can improve robustness by exploiting the complementary behavior of the two methods. Future work includes stronger pruning for phase polynomial synthesis and further improvements to the scalability of topology enumeration and solving.

\section*{Acknowledgment}

The authors have used ChatGPT to assist with manuscript writing. The work is partially supported by the Young Scientists Fund of the Natural Science Foundation of Jiangsu Province (No. BK20240536), and the
National Natural Science Foundation of China (No. 62502203).

% The authors would like to thank the project collaborators for discussions on exact synthesis, parity networks, and STP-based circuit solving.

\bibliographystyle{IEEEtran}
\bibliography{libpaper}

\vspace{-4em}
\begin{IEEEbiography}[{\includegraphics[width=1in,height=1.25in,clip,keepaspectratio]{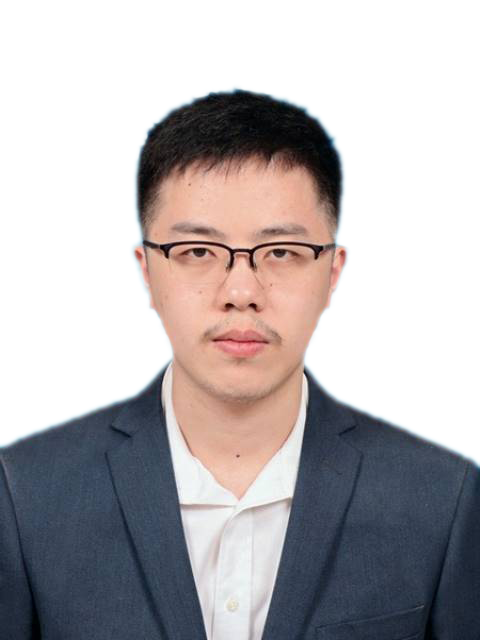}}]{Chenjian Li}
 received his B.Sc. degree in physics and computer science from Peking University in 2022. He is currently a Ph.D. student at the Institute of Software, Chinese Academy of Sciences. His research interests include the synthesis and optimization of quantum circuits and quantum differential privacy.
\end{IEEEbiography}
\vspace{-4em}

\begin{IEEEbiography}[{\includegraphics[width=1in,clip,keepaspectratio]{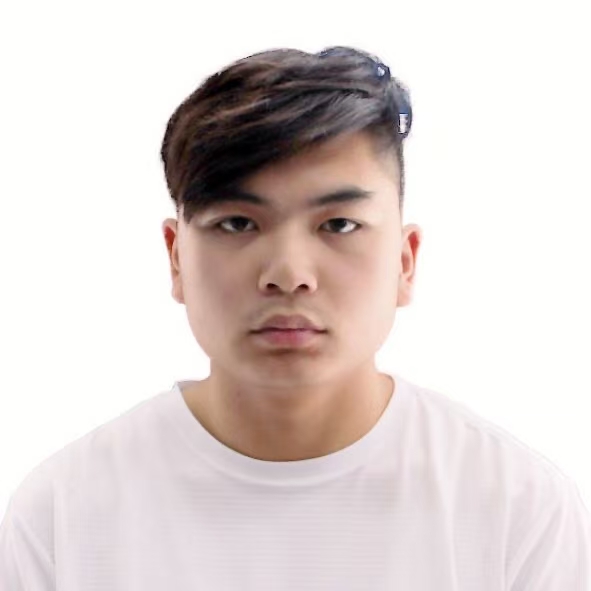}}]{Dingchao Gao} received his B.Sc. from Xidian University in 2021 and his M.S. from the Institute of Software, Chinese Academy of Sciences, in 2024. He is currently working as a software engineer at the Institute of Software, Chinese Academy of Sciences.

\end{IEEEbiography}
\vspace{-4em}
\begin{IEEEbiography}[{\includegraphics[width=1in,height=1.25in,clip,keepaspectratio]{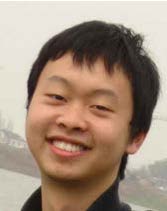}}]
{Xiangzhen Zhou} received the B.E. degree in engineering from Nanjing Normal University, Nanjing, China, in 2010, and the Ph.D. degree in engineering from Southeast University, Nanjing. From 2018 to 2020, he was a Visiting Student with the Centre for Quantum Software and Information, University of Technology Sydney, Ultimo, NSW, Australia. He is currently an Associate Professor with Nanjing Tech University, Nanjing. His research interests include quantum computing and quantum circuit optimization.
\end{IEEEbiography}
\vspace{-4em}

\begin{IEEEbiography}[{\includegraphics[width=1in,height=1.25in,clip,keepaspectratio]{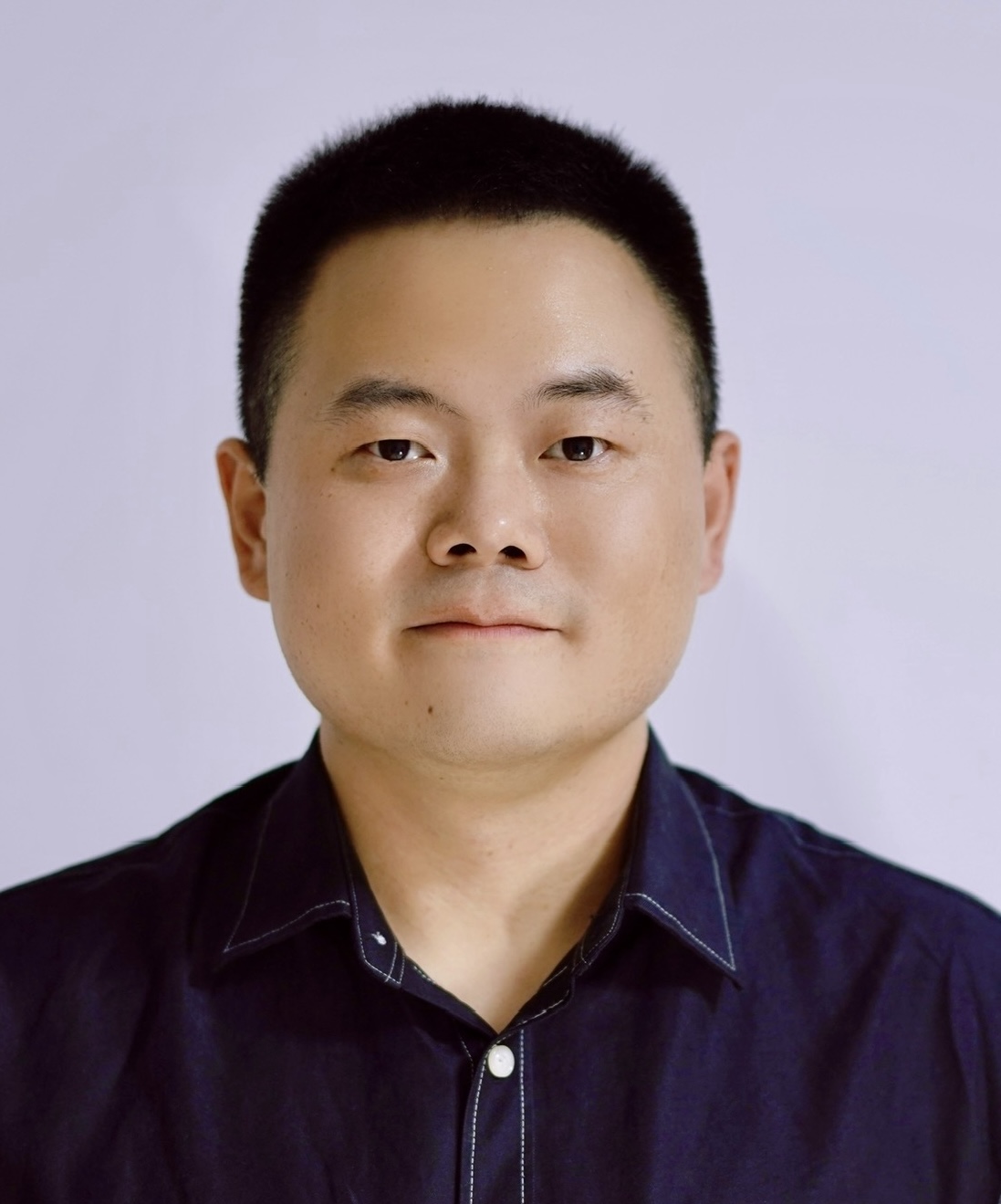}}]
{Ji Guan} received the B.Sc. degree in mathematics from Sichuan University, China, in 2014, and the Ph.D. degree in computer science from the University of Technology Sydney, Australia, in 2018. He is currently a Associate Research Professor at the Institute of Software, Chinese Academy of Sciences. His primary research interests are trustworthy quantum machine learning and model checking quantum systems.
\end{IEEEbiography}
\vspace{-4em}

\begin{IEEEbiography}[{\includegraphics[width=1in,height=1.25in,clip,keepaspectratio]{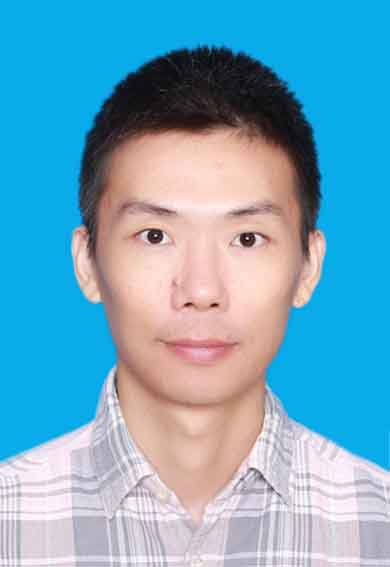}}]
{Pengcheng Zhu} received the B.E. and M.E. degrees in Computer Science and Technology from Liaoning Technical University, Fuxin, China, in 2004 and 2007, respectively, and the Ph.D. degree in Information and Communication Engineering from Nantong University, Nantong, China, in 2022.
He is currently an Associate Professor with the College of Artificial Intelligence, Taizhou University, China. His current research interests include reversible logic synthesis, quantum circuit synthesis, and quantum computing intelligence.
\end{IEEEbiography}
\vspace{-4em}

\begin{IEEEbiography}[{\includegraphics[width=1in,height=1.25in,clip,keepaspectratio]{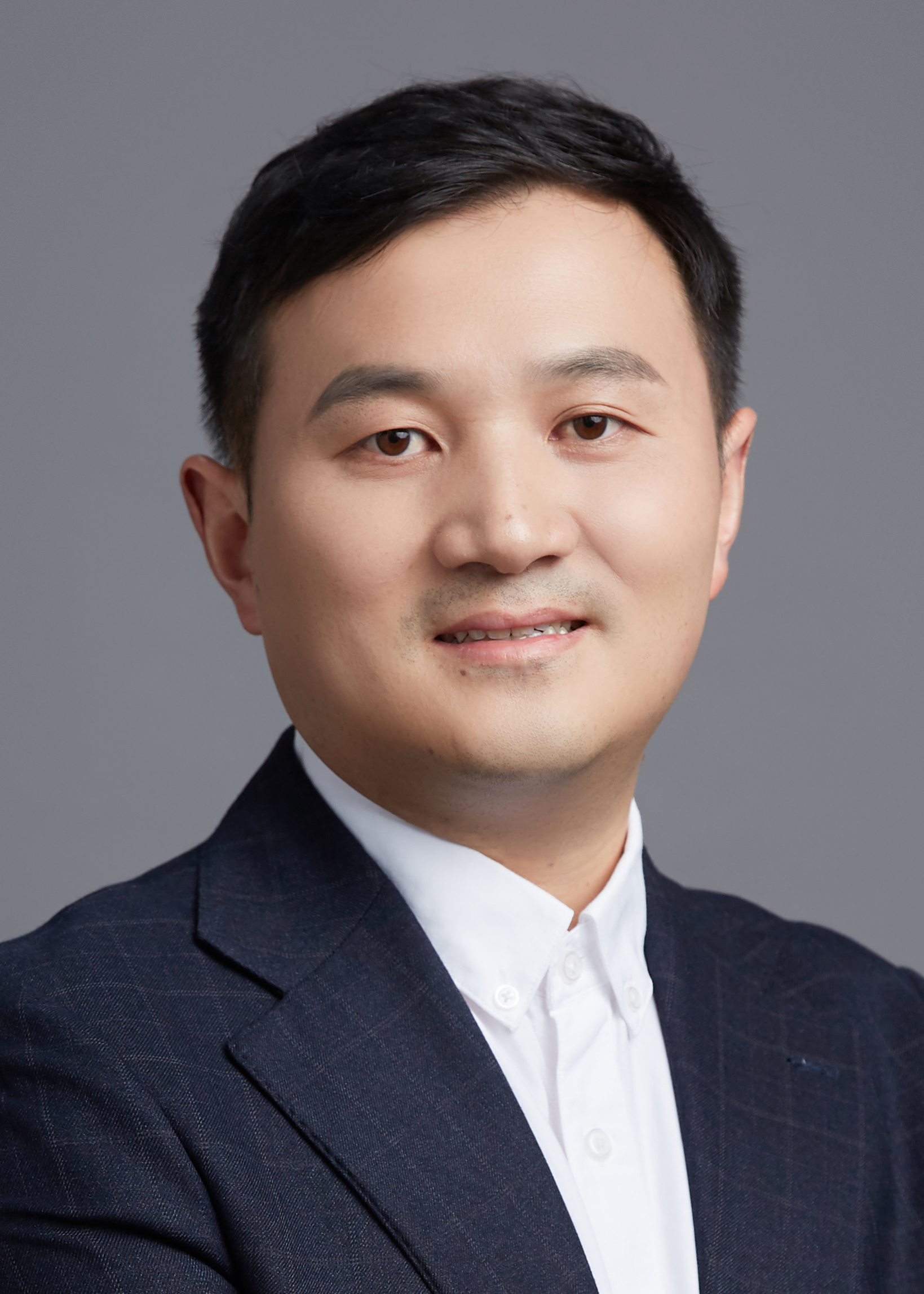}}]
{Zhufei Chu} (Member, IEEE)
received the B.S. degree in electronic engineering from Shandong University, Weihai, China, in 2008 and the M.S. and Ph.D. degrees in communication and information systems from the School of Integrated Circuits, Ningbo University, Ningbo, China, in 2011 and 2014, respectively. He was a Postdoctoral Fellow at the \'Ecole Polytechnique F\'ed\'erale de Lausanne (EPFL) from 2016 to 2017. He is currently a Professor with the School of Integrated Circuits, Ningbo University, Ningbo, China.

He serves as the General Chair (2026), Proceedings Chair (2019--2021), and Finance Chair (2022--2023) of the International Workshop on Logic and Synthesis (IWLS); Track Chair (2023--2025) of the IEEE International Symposium of EDA (ISEDA); and a Technical Program Committee Member for the Design Automation Conference (DAC), Design, Automation and Test in Europe Conference (DATE), Asia and South Pacific Design Automation Conference (ASP-DAC), and IWLS.

He actively maintains the logic synthesis framework ALSO (\url{https://github.com/nbulsi/also}). His current research interests include various aspects of logic synthesis and its applications.
\end{IEEEbiography}
\vspace{-4em}

%\appendices

% =============================================================================
% Supplementary material (generated from supp_mat.tex)
% =============================================================================
\clearpage
\onecolumn
\begingroup

% Notation used only by the supplementary material.
\newcommand\p{{p_{\v a}}}
\renewcommand\L{\mathsf{L}}
\newcommand\R{\mathsf{R}}
\newcommand\D{\mathsf{D}}
\newcommand\M{\mathcal{M}}

% Reproduce IEEEtran journal-title typography without a second \maketitle.
\begin{center}
  {\normalfont\Huge Supplementary Material for\\
``Parallelizable Exact Synthesis of Quantum Circuits via Semi-Tensor Product''\par}
  \vspace{1.5em}
  {\normalfont\sublargesize Chenjian Li, Dingchao Gao, Xiangzhen Zhou, Ji Guan,
Pengcheng Zhu, and Zhufei Chu\par}
\end{center}
\vspace{1.5em}

\appendices

% Supplement-only numbering: (A.1), Fig. A.1, Table A.1, and so on.
\numberwithin{equation}{section}
\makeatletter
\@addtoreset{figure}{section}
\@addtoreset{table}{section}
\makeatother
\renewcommand{\thefigure}{\thesection.\arabic{figure}}
\renewcommand{\thetable}{\thesection.\arabic{table}}

% -----------------------------------------------------------------------------
% Appendix 1
% -----------------------------------------------------------------------------
\section{Proof of Theorem 2}
\label{sec:supp_factorization_proof}

We first show that the structural matrices of parity Boolean functions exhibit a \emph{recursive symmetry}, then use the property to prove the theorem.

For any given Boolean function $\Phi$, we denote the left and right half of its structural matrix as
\begin{equation}
    \L[M_\Phi] = M_{\Phi|_{x_1=0}},\quad
    \R[M_\Phi] = M_{\Phi|_{x_1=1}}.
\end{equation}
For recursive subdivision, we would abbreviate as $\L\L[M_\Phi]:=\L[\L[M_\Phi]]$.

\begin{supplemma}
    The structural matrices of parity Boolean function is \emph{symmetric}. More specifically, for a parity Boolean function $\p(\v x)=\bigoplus_{i=1}^n a_ix_i$, either of the following condition holds:
    \begin{enumerate}
        \item if $a_1=0$, $\L[M_{\p}]=\R[M_{\p}]$,
        \item if $a_1=1$, $\L[M_{\p}]=\overline{\R[M_{\p}]}$,
    \end{enumerate}
    where $\overline{M}$ dentotes the element-wise negation of binary matrices.
\end{supplemma}
\begin{proof}
    If $a_1=0$, then $\p$ is independent of the value $x_1$. Then we have:
    \begin{equation}
        \L[M_{\p}]=M_{\p|_{x_1=0}}=M_{\p|_{x_1=1}}=\R[M_{\p}].
    \end{equation}
    If $a_1=1$, then
    \begin{align}
\p|_{x_1=1}
&=1\oplus \left(\bigoplus_{i=2}^n a_ix_i\right)\\
&=\lnot\left(\bigoplus_{i=2}^n a_ix_i\right)\\
&=\overline{\p|_{x_1=0}},
\end{align}
which leads to the relation on structural matrices:
\begin{equation}
\L[M_\p]=M_{\p|_{x_1=0}}=M_{\lnot (\p|_{x_1=1})}=\overline{\R[M_\p]}.
\end{equation}
Thus we have proved the symmetry of the structural matrices of parity Boolean functions.
\end{proof}

Then by applying the technique used in the proof repeatedly, we can further prove the \emph{recursive symmetry} of structural matrices of parity Boolean matrices, which, states that all $1/2^t$ portions of the structural matrix for a given parity Boolean functions, are somewhat similar and resemble each other.
\begin{suppdefinition}[Subdivided submatrices of structural matrix]
    Suppose $M_\Phi$ is a structural matrix for some Boolean function $\Phi$. Denote all 1-subdivded submatrix of $M_\Phi$ as $\mathcal M_{(1)}:=\{\L[M_{\Phi}],\R[M_\Phi]\}$. Then we can define the set of all $k$-subdivided submatrix of $M_\Phi$ as:
    \begin{align}
        \mathcal M_{(k)}&=\{\L[M],\R[M]: M\in\mathcal M_{(k-1)}\} \quad k\ge 2\\
        &=\{\mathsf{D_k}\mathsf{D_{k-1}}\cdots\mathsf{D_1}[M_f]:\mathsf{D}_i=\L\text{ or }\R\}
    \end{align}.
\end{suppdefinition}
\begin{supplemma}[Recursive symmetry]\label{lemma:recursive_sym}
    The structural matrix of a parities Boolean function is \emph{recursive symmetric}. More specifically, for $n$-variable parity Boolean function $\p$ and its associated structural matrix $M_\p\in \Ftwo^{2\times 2^n}$, there exists matrix $m\in \Ftwo^{2\times 2^{n-k}}$ for $k\le n$ such that all $k$-subdivided submatrix of $M_\p$ satisfy
    \begin{equation}
        M=m\text{  or  } M=\overline m,\quad \forall\,M\in \mathcal M_{(k)}.
    \end{equation}
\end{supplemma}

Finally we can prove the following theorem, by simplifying the general STP-factorization criterion (Theorem 1) with recursive symmetry of the structural matrices of parity Boolean functions.
\begin{supptheorem}
    Suppose $\Phi(\v x)=p_{\v a}(\v x)=\bigoplus_{i=1}^n a_ix_i$ is a parity Boolean function and $M_{p_{\v a}}=\big[M_1\,M_2\,\cdots\,M_{4\cdot 2^t}\big]$ is the structural matrix corresponding to $p_{\v a}$. Then factorization $M_{p_{\v a}}=M'\stp (I_{2^t}\otimes M_{\oplus})$ is feasible for $t\le n-2$ if and only if $M_1=M_4$ and $M_2=M_3$.
\end{supptheorem}
\begin{proof}
    We prove that it suffices to check $M_1\sim M_4$ to validating the STP factorization (``$\Leftarrow$''). The other direction (``$\Rightarrow$'') can be directly validated throught straightforward calculation and is trivial.

    Denote the first $t$-subdivied submatrix of $M_\p$ as $M_0^{(t)}:=\L^t[M_\p]=\big[M_1\, M_2\, M_3\, M_4]$. Then we see $(M_1=M_4\big)\land\big(M_2=M_3\big)$ can be abstracted as a condition on the first $t$-subdivided submatrix block:
    \begin{equation}
        \big(\L\L[M_0^{(t)}]=\R\R[M_0^{(t)}]\big)\land \big(\L\R[M_0^{(t)}]=\R\L[M_0^{(t)}]\big)
    \end{equation}
    By the recursive symmetry of $M_\p$, we see that the local equivalence condition propagates to all $t$-subdivded submatrices:
    \begin{equation}
        \big(\L\L[M]=\R\R[M]\big)\land \big(\L\R[M]=\R\L[M]\big), \quad M\in\mathcal M_{(t)}.
    \end{equation}
    Translating the relation back to relations of submatrices $M_i$s, we get
    \begin{equation}
        M_{4j+1}=M_{4j+4},\,
        M_{4j+2}=M_{4j+3},\quad \forall\, j\in[2^t].
    \end{equation}
    Then by the general STP-factorization criterion (Theorem 1), we see the factorization $M_{p_{\v a}}=M'\stp (I_{2^t}\otimes M_{\oplus})$ is feasible.
\end{proof}

\section{Heuristic used in factorizer}

In the factorization module, the factorizer estimates the factorization cost required for each equation to reach the next opportunity for reducing a gate decision, and prioritizes the equation with the lowest estimated cost.

For an equation that takes the form $\Phi=[\text{other factors}]\stp M_{g_i}\stp [M_w\text{ and } M_r\text{ factors}]=M$, the factorizer will not be able to eliminate gate decisions until the gate-related structural matrix $M_{g_i}$ is reached. Furthermore, we note that the cost of factorization steps related to $M_r$ and $M_w$ are much smaller compared to that of structural matrices related to gates. Denote $\Phi$ as $\Phi=M_L\stp M_{g_i}\stp M_R$ and $M=M'\stp M_R$, a rudimentary heuristic is then to estimate the cost of the closest gate-related factorization step:
\begin{equation}
    h_1(\Phi,M)=[\text{matrix width of }M']
\end{equation}

However, we note that factorization of $M_r$ factors would introduce extra indefinite variables, which would greatly complicate the following factorizing and inferring process. To account for the effect, we finally use the following recursively defined heuristic:
\begin{align}
    H_{\rm{fac}}(I,I)&:=\infty\\
    H_{\rm{fac}}(M_L\stp M_w, M)&:=H_{\rm{fac}}(M_L,M') \\
    H_{\rm{fac}}(M_L\stp M_r,M)&:=H_{\rm{fac}}(M_L,M')+2\cdot\#[\text{introduced var.}] \\
    &=H_{\rm{fac}}(M_L,M')+4\cdot [\text{matrix width of }M] \\
    H_{\rm{fac}}(M_L\stp M_{g_i},M)&:=[\text{matrix width of }M']
\end{align}
where $M'$ denotes the residual structural matrix of $M$ after the factorization step, and $\#[\text{introduced var.}]$ denotes the number of indefinite variables introduced in the factorization step. The first line corresponds to an empty equation $I=I$, which does not have actual contraining power. Consequently, its heuristic is set to infinity so its factorization priority is the lowest. Note that the concrete value of the heuristic only involves sum of matrix widths, so $M'$ does not have to be computed in advance to derive the heuristic.

\section{CNOT Gate Number Heuristic $H_{\rm CNOT}$}
In the peephole rewriting optimization experiment, the heuristic $H_{\rm CNOT}$ used to predict the number of CNOT gates required to implement a linear parity function $A$ is defined as follows:
\begin{align}
    f(M)&=\sum_{j}\log\left(1+\sum_i M_{ij}\right) \\
    h_{\rm CNOT}(A)&=f(A\oplus I_n)\\
    H_{\rm CNOT}(A)&=\frac{1}{4}\big[h_{\rm CNOT}(A)+h_{\rm CNOT}(A^\top)+h_{\rm CNOT}(A^{-1})+h_{\rm CNOT}(A^{-\top})\big]
\end{align}
where $H_{\rm CNOT}(A)$ is the final heuristic function. The intuition behind this heuristic is to estimate how ``far'' a linear function is from the identity matrix. The basic $f(\cdot)$ function prioritizes the elimination of columns that are ``almost done". Averaging over four equivalent variants of $A$ further refines the estimate, since these variants require exactly the same number of CNOT gates to implement as $A$ but provide more diversity. The heuristic function is adapted from the one proposed in \cite{debrugiere2021Gaussian}.

% -----------------------------------------------------------------------------
% Appendix 2
% -----------------------------------------------------------------------------
\newpage
\section{Parallel speedup of with different workers numbers on random CNOT circuits}
\label{sec:supp_paral_speedup}
\begin{figure}[ht]
    \centering
    \includegraphics[width=0.49\textwidth]{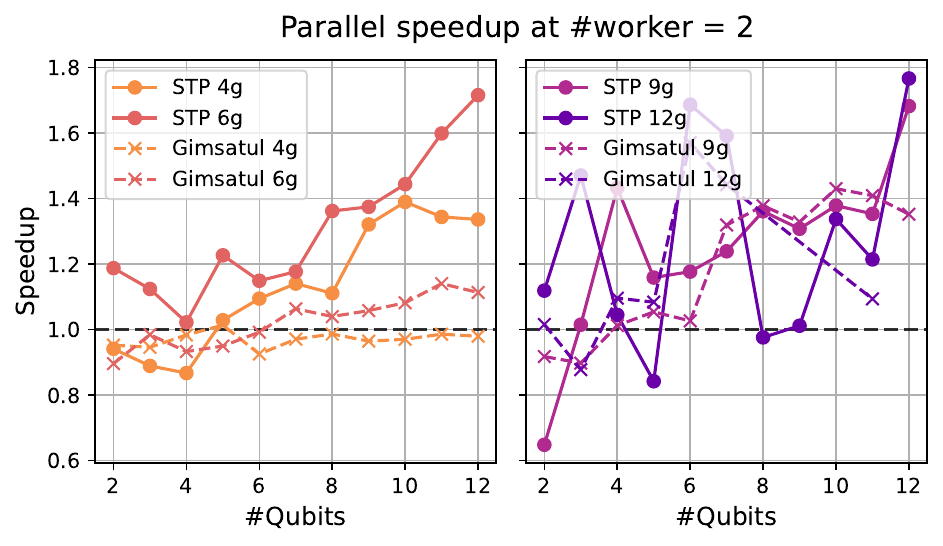}
    \includegraphics[width=0.49\textwidth]{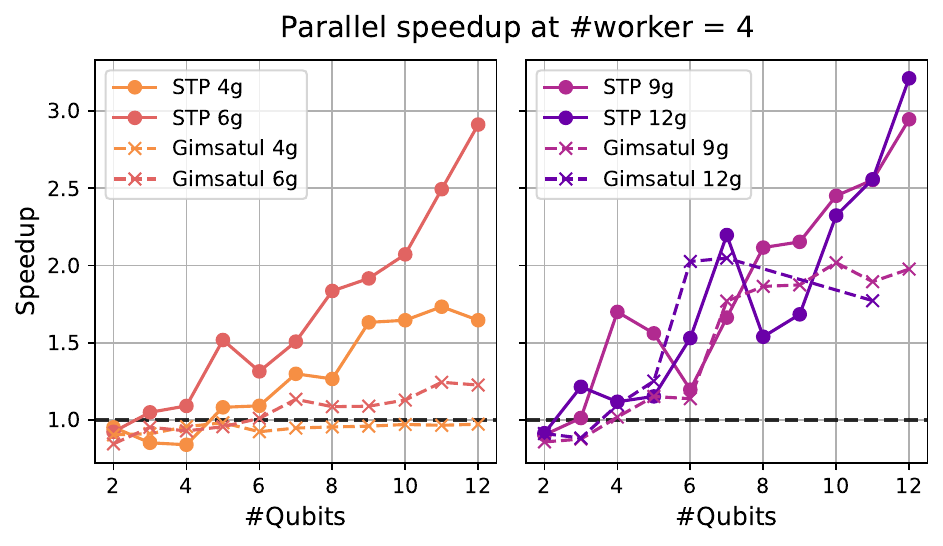}
    \includegraphics[width=0.49\textwidth]{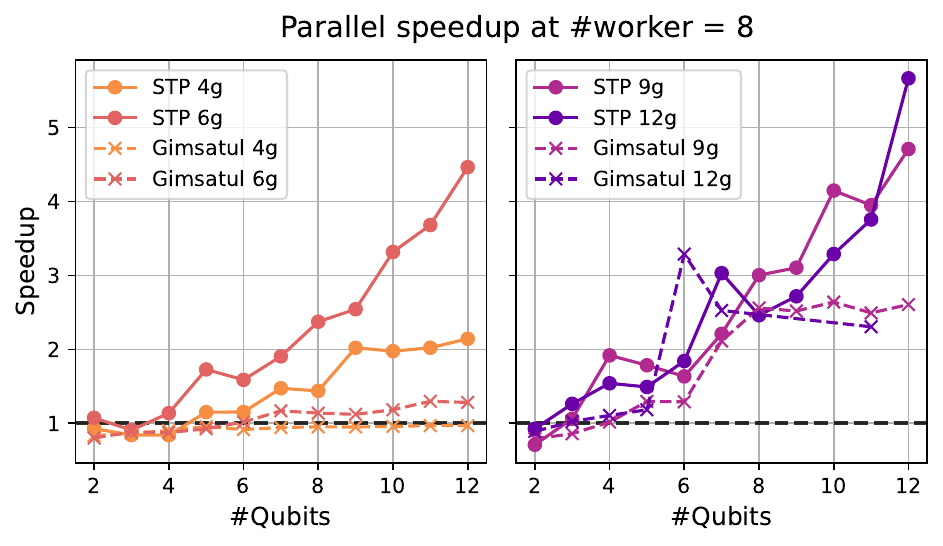}
    \includegraphics[width=0.49\textwidth]{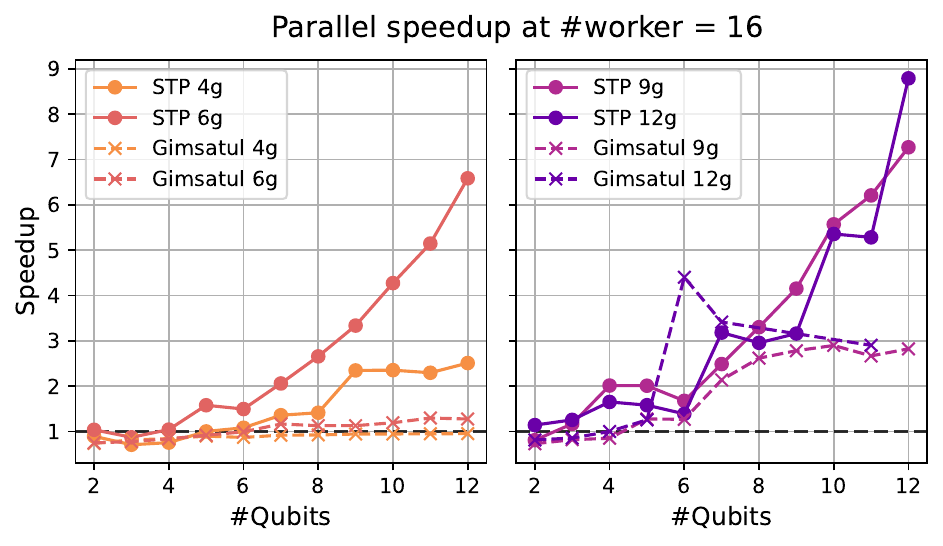}
    \includegraphics[width=0.49\textwidth]{figures_parallel_speedup_workers32.pdf}
    \caption{STP against Qiskit-SAT with the Gimsatul backend with different numbers of workers on random CNOT circuits. For each data point, 10 random circuits are generated and tested.}
    %\label{fig:label}
\end{figure}

% -----------------------------------------------------------------------------
% Appendix 3
% -----------------------------------------------------------------------------
\newpage
\section{Detailed data on circuit peephole rewriting optimization on QASMBench}
\label{sec:supp_qasmbench_table}

% Generated by exp/anal/build_phase_poly_table_v7_detailed.py.
% Requires: adjustbox, booktabs, makecell, and multirow.
% Generated by exp/anal/build_phase_poly_table_v7_detailed.py.
% Requires: adjustbox, booktabs, makecell, and multirow.
\begin{table*}[h]
\caption{Detailed circuit peephole rewriting optimization on QASMBench circuits with $n_g\le 100$.}
\label{table:detailed_qasmbench_v7_lt100}
\begin{adjustbox}{max width=\textwidth, center}
\begin{tabular}{lccccc|ccccc}
\toprule
\multirow{2}{*}[-0.8ex]{Case} &
\multirow{2}{*}[-0.8ex]{\#qubits} &
\multirow{2}{*}[-0.8ex]{\#gates} &
\multirow{2}{*}[-0.8ex]{\makecell{\#CNOT\\shards}} &
\multirow{2}{*}[-0.8ex]{\makecell{\#Phase Poly.\\shards}} &
\multirow{2}{*}[-0.8ex]{\#CNOT} &
\multirow{2}{*}[-0.8ex]{\makecell{Qiskit-SAT~\cite{qiskit-sat}\\Runtime($s$)}} &
\multirow{2}{*}[-0.8ex]{\makecell{Gimsatul~\cite{fleury2022gimsatul}\\Runtime($s$) (=1)}} &
\multirow{2}{*}[-0.8ex]{\makecell{\textbf{STP (ours)}\\Runtime($s$)}} &
\multicolumn{2}{c}{\textbf{STP+Gimsatul Portfolio}} \\
\cmidrule(lr){10-11}
& & & & & & & & & Runtime($s$) & STP wins \\
\midrule
cat\_state\_n4 & 4 & 4 & 1 & 0 & 3 & 1.4 (0.1$\times$) & 0.13 & 0.02 (5.9$\times$) & \textbf{0.02 (6.6$\times$)} & 1/1 (100\%) \\
qrng\_n4 & 4 & 4 & 0 & 0 & 0 & \textbf{0 (NA)} & \textbf{0} & \textbf{0 (NA)} & \textbf{0 (NA)} & 0/0 (NA) \\
deutsch\_n2 & 2 & 5 & 0 & 0 & 1 & 0.00040 (0.2$\times$) & 0.000065 & \textbf{0.000060 (1.1$\times$)} & 0.000061 (1.1$\times$) & 0/0 (NA) \\
teleportation\_n3 & 3 & 8 & 1 & 0 & 2 & 0.00040 (0.2$\times$) & 0.000070 & \textbf{0.000068 (1.0$\times$)} & 0.000077 (0.9$\times$) & 0/0 (NA) \\
iswap\_n2 & 2 & 9 & 0 & 0 & 2 & 0.00048 (0.2$\times$) & \textbf{0.000092} & 0.000098 (0.9$\times$) & 0.00012 (0.7$\times$) & 0/0 (NA) \\
lpn\_n5 & 5 & 11 & 1 & 0 & 2 & 0.00035 (0.2$\times$) & \textbf{0.000067} & 0.000085 (0.8$\times$) & 0.000070 (1.0$\times$) & 0/0 (NA) \\
qaoa\_n3 & 3 & 15 & 2 & 1 & 6 & 0.51 (0.6$\times$) & 0.29 & \textbf{0.04 (6.6$\times$)} & 0.05 (5.5$\times$) & 2/2 (100\%) \\
grover\_n2 & 2 & 16 & 0 & 0 & 2 & 0.00042 (0.2$\times$) & 0.000091 & 0.000091 (1.0$\times$) & \textbf{0.000090 (1.0$\times$)} & 0/0 (NA) \\
toffoli\_n3 & 3 & 18 & 1 & 1 & 6 & 0.40 (0.6$\times$) & 0.25 & \textbf{0.04 (6.4$\times$)} & 0.05 (5.1$\times$) & 1/1 (100\%) \\
fredkin\_n3 & 3 & 19 & 3 & 2 & 8 & 0.30 (0.7$\times$) & 0.22 & 0.04 (5.6$\times$) & \textbf{0.04 (6.1$\times$)} & 1/1 (100\%) \\
cat\_state\_n22 & 22 & 22 & 4 & 0 & 21 & 65.0 (0.04$\times$) & 2.7 & 0.22 (12.4$\times$) & \textbf{0.21 (13.3$\times$)} & 12/12 (100\%) \\
ghz\_state\_n23 & 23 & 23 & 5 & 0 & 22 & 83.6 (0.03$\times$) & 2.4 & 0.17 (14.3$\times$) & \textbf{0.16 (14.4$\times$)} & 9/9 (100\%) \\
adder\_n4 & 4 & 23 & 2 & 1 & 10 & 3.8 (0.1$\times$) & 0.53 & \textbf{0.16 (3.4$\times$)} & 0.16 (3.2$\times$) & 2/2 (100\%) \\
qec\_en\_n5 & 5 & 25 & 3 & 0 & 10 & 2.0 (0.1$\times$) & 0.28 & 0.04 (6.8$\times$) & \textbf{0.04 (7.3$\times$)} & 2/2 (100\%) \\
linearsolver\_n3 & 3 & 27 & 0 & 0 & 4 & 0.00057 (0.3$\times$) & \textbf{0.00015} & 0.00015 (1.0$\times$) & 0.00016 (0.9$\times$) & 0/0 (NA) \\
hs4\_n4 & 4 & 28 & 0 & 0 & 4 & 0.00049 (0.3$\times$) & 0.00014 & \textbf{0.00014 (1.0$\times$)} & 0.00014 (1.0$\times$) & 0/0 (NA) \\
wstate\_n3 & 3 & 34 & 1 & 3 & 9 & 0.53 (0.5$\times$) & 0.25 & \textbf{0.04 (6.0$\times$)} & 0.05 (5.2$\times$) & 1/1 (100\%) \\
cat\_n35 & 35 & 35 & 7 & 0 & 34 & 110 (0.04$\times$) & 4.3 & 0.34 (12.7$\times$) & \textbf{0.31 (13.6$\times$)} & 18/18 (100\%) \\
ghz\_n40 & 40 & 40 & 8 & 0 & 39 & 124 (0.03$\times$) & 4.3 & 0.35 (12.1$\times$) & \textbf{0.33 (13.1$\times$)} & 19/19 (100\%) \\
bv\_n14 & 14 & 41 & 0 & 0 & 13 & 0.00080 (0.4$\times$) & \textbf{0.00035} & 0.00039 (0.9$\times$) & 0.00035 (1.0$\times$) & 0/0 (NA) \\
quantumwalks\_n2 & 2 & 43 & 0 & 0 & 3 & 0.00049 (0.3$\times$) & 0.00013 & \textbf{0.00011 (1.2$\times$)} & 0.00011 (1.2$\times$) & 0/0 (NA) \\
simon\_n6 & 6 & 44 & 3 & 2 & 14 & 1.0 (0.5$\times$) & 0.55 & \textbf{0.18 (3.0$\times$)} & 0.20 (2.7$\times$) & 2/2 (100\%) \\
qec9xz\_n17 & 17 & 53 & 9 & 0 & 32 & 3.8 (0.1$\times$) & 0.31 & \textbf{0.04 (8.4$\times$)} & 0.04 (8.0$\times$) & 2/2 (100\%) \\
variational\_n4 & 4 & 54 & 0 & 8 & 16 & \textbf{0.0048 (7.6$\times$)} & 0.04 & 0.04 (1.0$\times$) & 0.04 (1.0$\times$) & 0/0 (NA) \\
bv\_n19 & 19 & 56 & 0 & 0 & 18 & 0.0011 (0.4$\times$) & \textbf{0.00045} & 0.00054 (0.8$\times$) & 0.00048 (0.9$\times$) & 0/0 (NA) \\
cat\_n65 & 65 & 65 & 13 & 0 & 64 & 269 (0.03$\times$) & 7.9 & 0.63 (12.6$\times$) & \textbf{0.61 (12.9$\times$)} & 34/34 (100\%) \\
bell\_n4 & 4 & 65 & 0 & 2 & 7 & 0.0022 (0.5$\times$) & 0.0010 & 0.00097 (1.0$\times$) & \textbf{0.00095 (1.1$\times$)} & 0/0 (NA) \\
bv\_n30 & 30 & 78 & 0 & 0 & 18 & 0.0010 (0.5$\times$) & 0.00049 & 0.00051 (1.0$\times$) & \textbf{0.00044 (1.1$\times$)} & 0/0 (NA) \\
ghz\_n78 & 78 & 78 & 16 & 0 & 77 & 216 (0.04$\times$) & 9.1 & 0.73 (12.6$\times$) & \textbf{0.67 (13.6$\times$)} & 38/38 (100\%) \\
vqe\_n4 & 4 & 89 & 0 & 0 & 9 & 0.00070 (0.4$\times$) & 0.00028 & \textbf{0.00027 (1.0$\times$)} & 0.00027 (1.0$\times$) & 0/0 (NA) \\
multiply\_n13 & 13 & 98 & 7 & 6 & 40 & 10.9 (0.1$\times$) & 1.3 & 0.23 (5.7$\times$) & \textbf{0.23 (5.8$\times$)} & 7/7 (100\%) \\
pea\_n5 & 5 & 98 & 7 & 4 & 42 & 13.9 (0.1$\times$) & 1.5 & 60.2 (0.03$\times$) & \textbf{1.2 (1.3$\times$)} & 6/7 (86\%) \\
\midrule
\textbf{Median speedup} &  &  &  &  &  & 0.2$\times$ & 1.0$\times$ & 3.0$\times$ & 2.7$\times$ &  \\
\bottomrule
\end{tabular}
\end{adjustbox}
\end{table*}

\begin{table*}[tb]
\caption{Detailed circuit peephole rewriting optimization on QASMBench circuits with $100< n_g\le 1{,}000$.}
\label{table:detailed_qasmbench_v7_100_to_1000}
\begin{adjustbox}{max width=\textwidth, center}
\begin{tabular}{lccccc|ccccc}
\toprule
\multirow{2}{*}[-0.8ex]{Case} &
\multirow{2}{*}[-0.8ex]{\#qubits} &
\multirow{2}{*}[-0.8ex]{\#gates} &
\multirow{2}{*}[-0.8ex]{\makecell{\#CNOT\\shards}} &
\multirow{2}{*}[-0.8ex]{\makecell{\#Phase Poly.\\shards}} &
\multirow{2}{*}[-0.8ex]{\#CNOT} &
\multirow{2}{*}[-0.8ex]{\makecell{Qiskit-SAT~\cite{qiskit-sat}\\Runtime($s$)}} &
\multirow{2}{*}[-0.8ex]{\makecell{Gimsatul~\cite{fleury2022gimsatul}\\Runtime($s$) (=1)}} &
\multirow{2}{*}[-0.8ex]{\makecell{\textbf{STP (ours)}\\Runtime($s$)}} &
\multicolumn{2}{c}{\textbf{STP+Gimsatul Portfolio}} \\
\cmidrule(lr){10-11}
& & & & & & & & & Runtime($s$) & STP wins \\
\midrule
basis\_test\_n4 & 4 & 110 & 2 & 20 & 46 & 13.9 (0.1$\times$) & 1.5 & 46.0 (0.03$\times$) & \textbf{1.4 (1.1$\times$)} & 7/10 (70\%) \\
error\_correctiond3\_n5 & 5 & 114 & 9 & 2 & 49 & 2.0 (0.2$\times$) & 0.31 & \textbf{0.11 (2.9$\times$)} & 0.12 (2.7$\times$) & 6/6 (100\%) \\
qpe\_n9 & 9 & 123 & 11 & 11 & 43 & 37.2 (0.1$\times$) & 4.9 & 181 (0.03$\times$) & \textbf{4.4 (1.1$\times$)} & 5/10 (50\%) \\
ghz\_n127 & 127 & 127 & 25 & 0 & 126 & 391 (0.04$\times$) & 16.3 & 1.3 (12.9$\times$) & \textbf{1.2 (13.2$\times$)} & 65/65 (100\%) \\
cat\_n130 & 130 & 130 & 26 & 0 & 129 & 406 (0.04$\times$) & 16.7 & 1.3 (12.9$\times$) & \textbf{1.3 (13.3$\times$)} & 67/67 (100\%) \\
adder\_n10 & 10 & 142 & 9 & 12 & 65 & 3.5 (0.3$\times$) & 1.1 & \textbf{0.34 (3.3$\times$)} & 0.34 (3.2$\times$) & 13/13 (100\%) \\
basis\_change\_n3 & 3 & 145 & 0 & 0 & 10 & 0.00071 (0.4$\times$) & 0.00027 & 0.00027 (1.0$\times$) & \textbf{0.00025 (1.1$\times$)} & 0/0 (NA) \\
wstate\_n27 & 27 & 157 & 5 & 0 & 52 & 82.4 (0.04$\times$) & 3.0 & 0.27 (11.0$\times$) & \textbf{0.26 (11.6$\times$)} & 14/14 (100\%) \\
bv\_n70 & 70 & 176 & 0 & 0 & 36 & 0.0018 (0.4$\times$) & \textbf{0.00074} & 0.00081 (0.9$\times$) & 0.00075 (1.0$\times$) & 0/0 (NA) \\
sat\_n7 & 7 & 180 & 8 & 12 & 60 & 4.7 (0.4$\times$) & 1.7 & \textbf{0.34 (5.0$\times$)} & 0.35 (4.8$\times$) & 11/11 (100\%) \\
wstate\_n36 & 36 & 211 & 7 & 0 & 70 & 112 (0.04$\times$) & 4.0 & 0.34 (11.9$\times$) & \textbf{0.30 (13.2$\times$)} & 17/17 (100\%) \\
vqe\_uccsd\_n4 & 4 & 220 & 22 & 20 & 88 & 14.5 (0.2$\times$) & 2.5 & \textbf{0.57 (4.4$\times$)} & 0.59 (4.2$\times$) & 24/24 (100\%) \\
knn\_n25 & 25 & 230 & 12 & 12 & 96 & 6.6 (0.4$\times$) & 2.4 & \textbf{0.40 (6.0$\times$)} & 0.48 (4.9$\times$) & 12/12 (100\%) \\
swap\_test\_n25 & 25 & 230 & 12 & 12 & 96 & 6.6 (0.3$\times$) & 2.3 & 0.43 (5.3$\times$) & \textbf{0.42 (5.5$\times$)} & 12/12 (100\%) \\
ghz\_state\_n255 & 255 & 255 & 51 & 0 & 254 & 691 (0.04$\times$) & 29.3 & 2.3 (12.8$\times$) & \textbf{2.0 (14.4$\times$)} & 115/115 (100\%) \\
cat\_n260 & 260 & 260 & 52 & 0 & 259 & 586 (0.05$\times$) & 29.1 & 2.2 (13.2$\times$) & \textbf{2.2 (13.5$\times$)} & 116/116 (100\%) \\
ising\_n26 & 26 & 280 & 12 & 13 & 50 & 2.9 (0.5$\times$) & 1.3 & \textbf{0.34 (4.0$\times$)} & 0.47 (2.8$\times$) & 12/12 (100\%) \\
bigadder\_n18 & 18 & 284 & 18 & 22 & 130 & 6.9 (0.4$\times$) & 2.4 & 0.78 (3.1$\times$) & \textbf{0.63 (3.9$\times$)} & 24/24 (100\%) \\
knn\_n31 & 31 & 287 & 15 & 15 & 120 & 8.0 (0.4$\times$) & 2.9 & 0.78 (3.7$\times$) & \textbf{0.77 (3.7$\times$)} & 15/15 (100\%) \\
qf21\_n15 & 15 & 311 & 39 & 31 & 115 & 81.0 (0.2$\times$) & 14.0 & 660 (0.02$\times$) & \textbf{12.5 (1.1$\times$)} & 16/31 (52\%) \\
qram\_n20 & 20 & 321 & 25 & 25 & 136 & 42.7 (0.2$\times$) & 6.6 & 73.5 (0.1$\times$) & \textbf{4.8 (1.4$\times$)} & 22/30 (73\%) \\
bv\_n140 & 140 & 352 & 0 & 0 & 72 & 0.0029 (0.5$\times$) & 0.0014 & \textbf{0.0014 (1.0$\times$)} & 0.0014 (1.0$\times$) & 0/0 (NA) \\
ising\_n34 & 34 & 368 & 16 & 17 & 66 & 3.9 (0.5$\times$) & 1.8 & \textbf{0.59 (3.1$\times$)} & 0.65 (2.8$\times$) & 16/16 (100\%) \\
knn\_n41 & 41 & 382 & 20 & 20 & 160 & 10.9 (0.4$\times$) & 3.9 & \textbf{0.73 (5.4$\times$)} & 1.0 (3.8$\times$) & 20/20 (100\%) \\
swap\_test\_n41 & 41 & 382 & 20 & 20 & 160 & 10.8 (0.4$\times$) & 4.1 & \textbf{0.71 (5.7$\times$)} & 1.1 (3.7$\times$) & 20/20 (100\%) \\
adder\_n28 & 28 & 424 & 26 & 33 & 195 & 267 (0.2$\times$) & 47.2 & 12.1 (3.9$\times$) & \textbf{5.1 (9.2$\times$)} & 38/44 (86\%) \\
dnn\_n2 & 2 & 450 & 0 & 12 & 42 & 0.02 (0.9$\times$) & \textbf{0.01} & 0.01 (1.0$\times$) & 0.01 (0.9$\times$) & 0/0 (NA) \\
wstate\_n76 & 76 & 451 & 15 & 0 & 150 & 222 (0.04$\times$) & 8.8 & 0.77 (11.4$\times$) & \textbf{0.74 (11.9$\times$)} & 39/39 (100\%) \\
ising\_n42 & 42 & 456 & 20 & 21 & 82 & 4.9 (0.4$\times$) & 2.0 & \textbf{0.57 (3.6$\times$)} & 0.62 (3.3$\times$) & 20/20 (100\%) \\
ising\_n10 & 10 & 480 & 20 & 25 & 90 & 4.9 (0.4$\times$) & 2.0 & \textbf{0.58 (3.5$\times$)} & 0.62 (3.3$\times$) & 20/20 (100\%) \\
qaoa\_n6 & 6 & 558 & 0 & 16 & 54 & 0.05 (0.5$\times$) & 0.02 & \textbf{0.02 (1.0$\times$)} & 0.02 (1.0$\times$) & 0/0 (NA) \\
multiplier\_n15 & 15 & 574 & 44 & 45 & 246 & 44.0 (0.2$\times$) & 8.1 & 20.0 (0.4$\times$) & \textbf{3.6 (2.2$\times$)} & 53/55 (96\%) \\
dnn\_n33 & 33 & 608 & 16 & 61 & 248 & 10.2 (0.4$\times$) & 4.1 & \textbf{1.2 (3.4$\times$)} & 1.3 (3.2$\times$) & 31/31 (100\%) \\
knn\_n67 & 67 & 629 & 33 & 33 & 264 & 18.3 (0.4$\times$) & 6.5 & \textbf{1.3 (5.1$\times$)} & 1.4 (4.7$\times$) & 33/33 (100\%) \\
sat\_n11 & 11 & 679 & 34 & 50 & 252 & 20.1 (0.4$\times$) & 7.2 & \textbf{1.7 (4.2$\times$)} & 1.9 (3.8$\times$) & 58/58 (100\%) \\
hhl\_n7 & 7 & 689 & 12 & 70 & 196 & 19.2 (0.2$\times$) & 3.0 & 77.6 (0.04$\times$) & \textbf{2.7 (1.1$\times$)} & 6/9 (67\%) \\
wstate\_n118 & 118 & 703 & 24 & 0 & 234 & 344 (0.04$\times$) & 14.2 & 1.3 (10.8$\times$) & \textbf{1.2 (12.0$\times$)} & 61/61 (100\%) \\
bv\_n280 & 280 & 712 & 0 & 0 & 152 & 0.01 (0.5$\times$) & 0.0030 & \textbf{0.0029 (1.0$\times$)} & 0.0030 (1.0$\times$) & 0/0 (NA) \\
ising\_n66 & 66 & 720 & 32 & 33 & 130 & 7.8 (0.5$\times$) & 3.5 & \textbf{0.95 (3.7$\times$)} & 1.0 (3.5$\times$) & 32/32 (100\%) \\
qugan\_n39 & 39 & 759 & 38 & 55 & 296 & 10.4 (0.4$\times$) & 3.9 & \textbf{0.94 (4.2$\times$)} & 1.1 (3.7$\times$) & 19/19 (100\%) \\
swap\_test\_n83 & 83 & 781 & 41 & 41 & 328 & 22.4 (0.4$\times$) & 8.3 & \textbf{1.6 (5.2$\times$)} & 1.8 (4.7$\times$) & 41/41 (100\%) \\
qft\_n18 & 18 & 783 & 136 & 45 & 306 & 226 (0.1$\times$) & \textbf{33.5} & 1921 (0.02$\times$) & 34.8 (1.0$\times$) & 8/40 (20\%) \\
dnn\_n51 & 51 & 959 & 25 & 97 & 392 & 16.9 (0.4$\times$) & 7.1 & \textbf{2.4 (2.9$\times$)} & 2.6 (2.8$\times$) & 49/49 (100\%) \\
adder\_n64 & 64 & 988 & 60 & 81 & 455 & 711 (0.2$\times$) & 139 & 30.2 (4.6$\times$) & \textbf{13.2 (10.6$\times$)} & 86/100 (86\%) \\
\midrule
\textbf{Median speedup} &  &  &  &  &  & 0.4$\times$ & 1.0$\times$ & 3.7$\times$ & 3.6$\times$ &  \\
\bottomrule
\end{tabular}
\end{adjustbox}
\end{table*}

\begin{table*}[tb]
\caption{Detailed circuit peephole rewriting optimization on QASMBench circuits with $1{,}000< n_g\le 10{,}000$.}
\label{table:detailed_qasmbench_v7_1000_to_10000}
\begin{adjustbox}{max width=\textwidth, center}
\begin{tabular}{lccccc|ccccc}
\toprule
\multirow{2}{*}[-0.8ex]{Case} &
\multirow{2}{*}[-0.8ex]{\#qubits} &
\multirow{2}{*}[-0.8ex]{\#gates} &
\multirow{2}{*}[-0.8ex]{\makecell{\#CNOT\\shards}} &
\multirow{2}{*}[-0.8ex]{\makecell{\#Phase Poly.\\shards}} &
\multirow{2}{*}[-0.8ex]{\#CNOT} &
\multirow{2}{*}[-0.8ex]{\makecell{Qiskit-SAT~\cite{qiskit-sat}\\Runtime($s$)}} &
\multirow{2}{*}[-0.8ex]{\makecell{Gimsatul~\cite{fleury2022gimsatul}\\Runtime($s$) (=1)}} &
\multirow{2}{*}[-0.8ex]{\makecell{\textbf{STP (ours)}\\Runtime($s$)}} &
\multicolumn{2}{c}{\textbf{STP+Gimsatul Portfolio}} \\
\cmidrule(lr){10-11}
& & & & & & & & & Runtime($s$) & STP wins \\
\midrule
ising\_n98 & 98 & 1072 & 48 & 49 & 194 & 11.9 (0.4$\times$) & 5.3 & \textbf{1.5 (3.6$\times$)} & 1.7 (3.0$\times$) & 48/48 (100\%) \\
swap\_test\_n115 & 115 & 1085 & 57 & 57 & 456 & 32.0 (0.4$\times$) & 11.5 & \textbf{2.5 (4.7$\times$)} & 2.7 (4.2$\times$) & 57/57 (100\%) \\
knn\_129 & 129 & 1218 & 64 & 64 & 512 & 34.2 (0.4$\times$) & 13.0 & 3.6 (3.6$\times$) & \textbf{3.0 (4.3$\times$)} & 64/64 (100\%) \\
qugan\_n71 & 71 & 1415 & 70 & 103 & 552 & 21.9 (0.4$\times$) & 7.9 & \textbf{2.6 (3.1$\times$)} & 2.8 (2.8$\times$) & 35/35 (100\%) \\
adder\_n118 & 118 & 1834 & 111 & 153 & 845 & 1422 (0.2$\times$) & 245 & 56.1 (4.4$\times$) & \textbf{26.3 (9.3$\times$)} & 158/184 (86\%) \\
dnn\_n8 & 8 & 2032 & 0 & 60 & 192 & 0.61 (0.9$\times$) & 0.53 & 0.53 (1.0$\times$) & \textbf{0.53 (1.0$\times$)} & 0/0 (NA) \\
qft\_n29 & 29 & 2059 & 323 & 100 & 812 & 402 (0.2$\times$) & \textbf{83.8} & \textgreater 3600 & 87.7 (1.0$\times$) & 19/95 (20\%) \\
qugan\_n111 & 111 & 2235 & 110 & 163 & 872 & 28.6 (0.5$\times$) & 15.4 & 7.1 (2.2$\times$) & \textbf{7.0 (2.2$\times$)} & 55/55 (100\%) \\
wstate\_n380 & 380 & 2275 & 76 & 0 & 758 & 725 (0.1$\times$) & 41.1 & 3.1 (13.1$\times$) & \textbf{2.9 (14.1$\times$)} & 140/140 (100\%) \\
vqe\_uccsd\_n6 & 6 & 2282 & 274 & 150 & 1052 & 525 (0.1$\times$) & 58.7 & \textbf{12.0 (4.9$\times$)} & 13.0 (4.5$\times$) & 312/312 (100\%) \\
basis\_trotter\_n4 & 4 & 2490 & 56 & 150 & 582 & 27.2 (0.3$\times$) & 7.3 & 49.1 (0.1$\times$) & \textbf{4.8 (1.5$\times$)} & 57/58 (98\%) \\
gcm\_h6 & 13 & 3148 & 92 & 247 & 762 & 32.6 (0.6$\times$) & 21.0 & \textbf{13.8 (1.5$\times$)} & 14.2 (1.5$\times$) & 82/82 (100\%) \\
knn\_341 & 341 & 3232 & 170 & 170 & 1360 & 98.2 (0.4$\times$) & 40.9 & 16.8 (2.4$\times$) & \textbf{15.3 (2.7$\times$)} & 170/170 (100\%) \\
swap\_test\_n361 & 361 & 3422 & 180 & 180 & 1440 & 105 (0.4$\times$) & 44.3 & \textbf{18.6 (2.4$\times$)} & 18.6 (2.4$\times$) & 180/180 (100\%) \\
dnn\_n16 & 16 & 4064 & 0 & 124 & 384 & 9.4 (0.4$\times$) & 3.9 & 3.9 (1.0$\times$) & \textbf{3.8 (1.0$\times$)} & 0/0 (NA) \\
ising\_n420 & 420 & 4614 & 209 & 210 & 838 & 46.2 (0.7$\times$) & 31.5 & \textbf{15.8 (2.0$\times$)} & 17.1 (1.8$\times$) & 209/209 (100\%) \\
multiplier\_n45 & 45 & 5981 & 429 & 453 & 2574 & 332 (0.5$\times$) & 160 & 125 (1.3$\times$) & \textbf{103 (1.6$\times$)} & 600/605 (99\%) \\
adder\_n433 & 433 & 6769 & 408 & 573 & 3120 & 5125 (0.2$\times$) & 1144 & 340 (3.4$\times$) & \textbf{233 (4.9$\times$)} & 587/683 (86\%) \\
qugan\_n395 & 395 & 8057 & 394 & 589 & 3144 & 310 (0.8$\times$) & 243 & 217 (1.1$\times$) & \textbf{216 (1.1$\times$)} & 197/197 (100\%) \\
qft\_n63 & 63 & 9828 & 901 & 270 & 3906 & 2998 (0.1$\times$) & \textbf{317} & \textgreater 3600 & 335 (0.9$\times$) & 24/265 (9\%) \\
\midrule
\textbf{Median speedup} &  &  &  &  &  & 0.4$\times$ & 1.0$\times$ & 2.4$\times$ & 2.3$\times$ &  \\
\bottomrule
\end{tabular}
\end{adjustbox}
\end{table*}

\clearpage
\endgroup
\end{document}